\documentclass{llncs}
\usepackage[style=alphabetic]{biblatex}
\usepackage{kbordermatrix}
\usepackage{amsfonts,color,morefloats}
\usepackage{amssymb, amsmath,latexsym}
\usepackage{graphicx}
\newcommand{\Dlambdabeta}{\underline {\Lambda}_\beta}

\newcommand{\F}{\mathbb F}
\newcommand{\B}{\mathbb B}
\newcommand{\Tr}{\operatorname{Tr}}
\newcommand{\rank}{\operatorname{rank}}
\newcommand{\set}[1]{\left\{#1\right\}}

\newcommand{\RS}{\text{RS}}

\newcommand{\Cbase}{C_0}
\newcommand{\Gbase}{G_0}
\newcommand{\Hbase}{H_0}
\newcommand{\nbase}{n_0}
\newcommand{\kbase}{k_0}
\newcommand{\Dbase}{D_0}

\newcommand{\Cadd}{C}
\newcommand{\Dadd}{D}
\newcommand{\Gadd}{G}
\newcommand{\Hadd}{H}
\newcommand{\Daddbeta}{\underline{\Dadd}_{\beta}}
\newcommand{\Cobs}{C_{\rm comp}}
\newcommand{\Gobs}{G_{\rm comp}}
\newcommand{\Hobs}{H_{\rm comp}}

\newcommand{\Clin}{C}
\newcommand{\Dlin}{D}
\newcommand{\Glin}{G}
\newcommand{\Hlin}{H}
\newcommand{\Dlinpb}{D_\beta}

\title{On the Leakage of Massey Secret Sharing Schemes under Linear Computations}
\author{Nadja Aoutouf\inst{1,2,3} \and Daniel Augot\inst{1,2,3}}

\authorrunning{Nadja Aoutouf and Daniel Augot}

\institute{
INRIA                  % 1
\and
École Polytechnique    % 2
\and
LIX, CNRS UMR 7161     % 3
}

\begin{document}

\maketitle
\thispagestyle{plain}

% Provide the keywords *before* the abstract
% When keywords contain macros provide the text version as the optional argument

% Provide the abstract of your paper

\begin{abstract}

  Leakage attacks on secret sharing schemes exploit partial information about
  individual shares to recover the underlying secret. In coding
  theory, linear exact repair schemes (LERSs) enable the recovery of one
  codeword symbol from a small amount of information obtained from the remaining symbols, 
  provided that the code has sufficiently low rate. This can be interpreted as recovering 
  the secret from partial information, namely subfield symbols, of the shares.

  Recently, a randomized construction based on subfield subcodes was
  proposed for constructing LERS-derived leakage attacks against
  Massey secret sharing schemes based on general linear codes. 
  We extend this framework to multiple shared secrets whose corresponding 
  shares are related through linear computations, with leakage also 
  allowed on the computation outcomes.

  More precisely, we consider $N$ secrets, of which $K<N$ are linearly independent input values and the remaining $N-K$ secrets 
  are determined by linear computations on these inputs. We analyse the existence of LERS-derived leakage  that exploits this
  structure. We first study the case of addition  and then generalize our construction to arbitrary linear computations.

  Our analysis applies to general linear codes of length $n+1$ and dimension $k$ over 
  $\F_{q^m}$  with $k \leq Nn/(Km)$, and 
  supports arbitrary linear computations, whereas the previous subfield subcode construction only applies to $k \leq n/m-1.$ 
  Consequently, exploiting the linear relations enables LERS based leakage  which extend the range of code parameters vulnerable to such attacks.
  Finally, identical leakage functions can arise for certain linear relations, 
  making this a more realistic yet still potentially powerful attack model.
  Finally, simulations indicate that identical leakage functions can be used for certain linear relations, 
  yielding a more realistic attack model.

	% Linear exact repair schemes (LERSs) enable the reconstruction of a codeword
	% symbol of a linear code by downloading only subfield symbols from the
	% remaining symbols. Besides their applications to distributed storage,
	% LERSs have a cryptographic interpretation: they correspond to leakage attacks
	% on Massey secret sharing schemes. Recently, a randomized construction of
	% LERSs based on subfield subcodes was proposed for general linear codes.
	% In this work, we extend this framework from the recovery of a single secret
	% to the recovery of multiple secrets related by linear computations. More
	% precisely, we consider $N$ secrets satisfying a linear relation, where
	% $K<N$ of the secrets are linearly independent and the remaining secrets are
	% determined by the relation. Building on the subfield-subcode construction,
	% we derive explicit leakage attacks that exploit this additional linear
	% structure. For this, we
	% first study additive relations and then generalize our construction to
	% arbitrary linear computations represented by linear circuits.
	% Our construction applies to linear codes of length $n+1$ and dimension $k$
	% over $\F_{q^m}$ satisfying the dimensional threshold
	% $k \leq Nn/(Km)$. Consequently, exploiting the linear relations enables
	% leakage attacks for a larger dimensional threshold than the previous
	% subfield-subcode construction.

\end{abstract}

\section{Introduction}
Surprisingly, linear secret sharing schemes can themselves be subject to leakage attacks. 
In the leakage model for such attacks, each share is associated with a leakage function, such that the adversary 
obtains partial information about each share rather than the entire share. 
For example, an attacker may physically place a wire on an embedded device to measure a 
value of interest~\cite{TR:AK96}. Such attacks pose a serious threat to cryptographic 
implementations which use secret sharing (masking) to protect sensitive data.
These concerns initiated a line of work, beginning
with~\cite{C:BDIR18} and culminating
in~\cite{TCC:Kasser24,EC:Nguyen25}, with the goal of studying the
leakage resilience threshold of Shamir's secret sharing scheme.

In these works, the interplay between various parameters, such as the 
\emph{threshold}, the amount of information provided by the leakage, 
and the extension degree, is studied. The authors determine ranges of 
these parameters that guarantee resilience against arbitrary leakage 
from each share, regardless of the specific leakage functions chosen by the adversary.

From a constructive coding-theoretic perspective, viewed as an attack in cryptanalysis, 
Guruswami and Wootters~\cite{STOC:GurWoo16} construct a
linear exact repair scheme for full-support Reed--Solomon codes,
which directly yields leakage functions for Shamir's secret sharing scheme. More
precisely, consider a Reed--Solomon code of length $n+1$ over
$\F_{q^m}$ with evaluation points ${\alpha_0,\ldots,\alpha_n}$, such that 
the secret is identified with the first codeword symbol $c_0$, while
the adversary obtains, from each share $c_j$, $j\in[n]$, the leakage
$\Tr_{\F/\B}(c_j/\alpha_j)\in\F_q.$ In~\cite{STOC:GurWoo16} it is shown that
these leaked subsymbols are sufficient to reconstruct $c_0$ for
Reed--Solomon codes of dimension less than $(n+1)(1-1/q)$. In
particular, when $q=2$, the secret can be recovered from one bit
of leakage from each share.

The construction of explicit leakage
functions beyond the Reed--Solomon setting is considerably less
understood. A recent approach based on subfield
subcodes~\cite{wcc-paper} provides a general method for
constructing leakage functions for secret sharing schemes induced by
arbitrary linear codes. Given a linear code $C\subseteq\F_{q^m}^{n+1}$
and its dual $C^\perp$, this approach gives a randomized construction of 
leakage functions that enable secret recovery whenever
the dimension of $C$ is less than or equal to $n/m-1$. We refer to this result as the \emph{base case}.

\subsection{Motivation}
Massey's secret sharing scheme generalizes Shamir's secret sharing 
scheme by replacing Reed--Solomon codes with general codes~\cite{massey-secret-sharing93}. 
In cryptography, for instance in secure multiparty computation or in masking for protecting implementations, 
secrets can be processed through computations on their shares. For example, 
in threshold signature schemes, secret keys are distributed among several parties~\cite{BPR:PRATS22}. 
LERS-derived leakage resilience of secret sharing schemes has been mainly studied 
with respect to the recovery of a single secret from partial 
information leaked from its shares. The analysis of such leakage 
 when some computations involve shared secrets
has received significantly less attention.
This raises the question: can an 
adversary exploit this additional computational structure to recover the 
original secrets? In this work, we investigate this question for Massey's 
secret sharing schemes based on linear codes, considering, as a first step, linear computations.

\subsection{Contribution}
Massey's secret sharing schemes induced by a linear code $C\subseteq\F_{q^m}^{n+1}$ naturally preserve linear computations on secrets. 
For example, if secrets $u,v,w\in\F_{q^m}$ 
satisfy the additive relation $u+v=w$, 
then the corresponding maskings $c_u,c_v,c_w\in C$ of the secrets, such that $u=c_{u,0}$, $v=c_{v,0}$, and $w=c_{w,0}$, satisfy 
$c_u+c_v=c_w$. In this paper, we study the LERS based leakage in such settings. Our main contributions are as follows.
\begin{itemize}
    \item \textbf{Leakage attacks for additive relation.}
    % Building on the subfield subcode construction of leakage functions for Massey secret sharing schemes 
    %     introduced in~\cite{wcc-paper}, 
      We show that a single additive relation $w=u+v$ for secrets $u$,
      $v$ and $w$ enables leakage attacks beyond the base case
      of~\cite{wcc-paper}: for $k\leq 3n/(2m),$ there exist, with high probability, 
	  leakage functions that recover the secret from only $3n$ subsymbols over $\F_q$, 
	  while the base case (without computation) 
	  is restricted to $k\leq n/m-1.$ 
    \item \textbf{Experimental validation.}
    We validate our theoretical findings through simulations for the case of addition.
  \item \textbf{Generalization to linear computations}  For arbitrary
    linear computations with $N$ secrets, of which $K$ are linearly
    independent, and the remaining $N-K$ secrets are obtained by
    linear computations, we derive the bound
    $k\leq Nn/(Km)$.
	\item \textbf{Use of identical leakage functions.}
	We prove that the simple additive relation cannot be exploited when the same leakage functions are reused. 
	More precisely, any such repair scheme reduces to a repair scheme for the corresponding base code, 
	yielding no improvement over the base case. However, simulations indicate that for a general linear relation 
	$\mu_u u + \mu_v v = w$ with $\mu_u\neq\mu_v \in \F\setminus\{0,1\}$, the same leakage functions 
	can be reused while still exploiting the linear relation, making this a more realistic leakage attack.
\end{itemize}

% \subsection{Organization}
% Section 2 introduces the notation, definitions, and preliminary results used
% throughout the paper. 
% In Section 3, we present and analyze the construction of a leakage attack based on subfield
% subcodes for Massey secret sharing schemes with additive relations.
% Section 4 extends this construction to arbitrary linear
% circuits. Finally, Section 5 concludes the paper.

\section{Preliminaries}
For a positive integer $n$, we use the notation
$[n]_0 := \{0,1,\ldots,n\}$ and
$[n] := \{1,\ldots,n\}.$
We denote by $\B=\mathbb{F}_q$ the finite field of size $q$ and by
$\F$ an extension field of $\B$ of degree $m$.
We consider linear codes over $\F$ of length $n+1$. A
$[n+1,k]_{\F}$ code (we omit the minimum distance $d$, as it will not be
needed in the following) is a $k$-dimensional $\F$-linear subspace of
$\F^{n+1}.$ Whenever convenient, we also write $k(C)$ (resp. $n(C)$) for the dimension (resp. length) of $C$. 
Furthermore, we denote the redundancy of $C$ by $r(C)=n(C)-k(C)$ and the dual of $C$ by $D$.

\subsection{Puncturing and Shortening}
Let $C$ be a $[n+1,k]_{\F}$ code and let $I\subseteq [n]_0$ be an
index set. The \emph{puncturing} of $C$ at $I$, denoted by
$\underline{C}_I$, is the code of length $n+1-|I|$ obtained by
deleting the coordinates indexed by $I$, i.e.,
$\underline{C}_I
:=
\left\{
(c_j)_{j\in [n]_0\setminus I}
\;\middle|\;
(c_0,\ldots,c_n)\in C
\right\}.$
The \emph{shortening} of $C$ at $I$, denoted by $\overline{C}_I$,
is the code obtained by restricting to codewords that vanish on the
coordinates in $I$ and subsequently deleting these coordinates, i.e.,
$\overline{C}_I
:=
\left\{
(c_j)_{j\in [n]_0\setminus I}
\;\middle|\;
(c_0,\ldots,c_n)\in C,\;
c_i=0\text{ for all }i\in I
\right\}.$
It is well known that $\overline{(C^\perp)}_I=(\underline{C}_I)^\perp$. Moreover, if $H$ is a parity-check matrix of $C$, 
then a parity-check matrix of $\overline{C}_I$ is obtained by deleting from $H$ the columns indexed by $I$.
For readability, we will often  omit the index set $I$ and simply write
$\underline{C}$ and $\overline{C}$.

\subsection{Massey Secret Sharing Schemes}

We use the well-established connection between linear codes and secret sharing schemes~\cite{massey-secret-sharing93}, 
and follow the terminology of~\cite{doron2026discrepancyrandomlinearcodes}. We refer to a $[\nbase+1,\kbase]_{\F}$ 
code $\Cbase$ as the \emph{base code}, which we assume admits a generator matrix $\Gbase$ of the form
\begin{equation} 
	\Gbase= \begin{bmatrix} -1 & \phi\\ 
		0 & \overline{\Gbase} \end{bmatrix}, 
              \label{eq:Gbase}
\end{equation} 
where $\phi\in\F^\nbase$ is nonzero and $\overline{\Gbase}$ is a generator matrix of the shortening of $\Cbase$ at the first coordinate. 
Every $\underline{h}\in\underline{\Dbase}$ admits a unique extension $(h_0,\underline{h})\in\Dbase$, where
$h_0=\phi\cdot\underline{h}^{T}$.
We refer to the $\B$-linear map
	\begin{equation}\label{eq:phi:D}
			\phi_{\Dbase} :\begin{array}[t]{rcl}
				\underline{\Dbase}& \longrightarrow &\F \\
				\underline{h}& \longmapsto&  \phi\cdot\underline{h}^{T}
			\end{array}
    \end{equation}
as the \emph{extension map} associated with $\Dbase$.
Note that the assumption $\phi \neq 0^\nbase$ ensures that $\phi_{\Dbase}$ is nontrivial.
Let $s\in\F$ be a secret. 
The dealer samples a random codeword 
$c=(c_0,\ldots,c_{\nbase})\in\Cbase$
with the constraint $c_0=s$. The corresponding shares are  $c_1,\ldots,c_{\nbase}$. 
% Let $d_0$ and $d_0^\perp$ denote the minimum distances of the base code $\Cbase$ and the dual code $\Dbase$, respectively. 
% Then any set of $d_0^\perp-2$ shares provides no information about the secret, 
% whereas any set of $\nbase-d_0+2$ shares allow to reconstruct it.

\subsection{Subfield Subcodes}

Let $C\subseteq\F^{n+1}$ be a linear code. Its \emph{subfield
subcode} with respect to $\B$ is defined as
$C_\B := C\cap\B^{n+1}.$
Let $\zeta=(\zeta_1,\ldots,\zeta_m)$ be a basis of $\F$ over $\B$.
For a matrix $M\in\F^{u\times v}$, we denote by
$M_\B\in\B^{um\times v}$ the matrix obtained by expanding each
entry of $M$ with respect to $\zeta$, column-wise. More precisely,
the rows of $M_\B$ are indexed by pairs $(i,\ell)$, and
$(M_\B)_{(i,\ell),j}=(M_{ij})_\ell,$
where $(M_{ij})_\ell$ denotes the $\ell$-th coordinate of $M_{ij}$
with respect to $\zeta$. Similarly, we denote by $M^\B\in\B^{u\times
vm}$ the row-wise expansion of $M$.
If $H\in\F^{(n+1)\times(n+1-k)}$ is a parity-check matrix of $C$,
then $H_\B$ is a (not
necessarily full rank) parity-check matrix of $C_\B$. This gives
\begin{equation}
\dim_{\B}(C_\B)
\geq n+1-m(n+1-k).\label{eq:bound:subfield}
\end{equation}

\subsection{Linear Exact Repair Scheme by Guruswami and Wootters}
In a linear exact repair scheme (LERS), an erased codeword symbol is recovered from partial
information obtained from the remaining symbols, rather than by downloading the full symbols. 
In~\cite{STOC:GurWoo16}, Guruswami and Wootters showed that this approach can be used to repair an 
erased symbol of a Reed--Solomon codeword.
The determination of a secret from leakage in a Massey secret sharing scheme can be 
viewed as a linear exact repair problem: the secret coordinate $c_0$ is recovered 
from partial information obtained from the remaining code symbols.
A simplified version of a LERS is formalized as
follows.
\begin{definition}
	A \emph{linear exact repair scheme} (LERS) for a
	$[n+1,k]_{\F}$ code $C$ with respect to a subfield
	$\B\subseteq\F$ consists of
	\begin{itemize}
		\item $\B$-linear maps $g_j\colon\F\rightarrow\B$ for
		$j\in[n]$,
		\item a $\B$-linear reconstruction map
		$R\colon\B^n\rightarrow\F$,
	\end{itemize}
	such that, for every $c=(c_0,\ldots,c_n)\in C$, we have
	$c_0=R\bigl(g_1(c_1),\ldots,g_n(c_n)\bigr).$
\end{definition}
The main result of Guruswami and Wootters shows that full-length
Reed--Solomon codes over $\F$ admit a LERS over $\B$ when their
dimension is sufficiently small.
\begin{theorem}[\cite{STOC:GurWoo16}]
	Let $[\F:\B]=m$, $n+1=q^m$, and
	$k\leq (1-\frac{1}{q})(n+1)$. Then the Reed--Solomon code
	$\RS[\alpha,k]$ of length $n+1$, with support equal to the whole
	field $\F$, admits a LERS over $\B$.
\end{theorem}
More generally, \cite{STOC:GurWoo16} gives a criterion for the existence of a LERS.
\begin{theorem}[\cite{STOC:GurWoo16}]
	Let $\Cbase$ be a base $[\nbase+1,\kbase]_{\F}$ code. If there
	exist $m$ dual codewords
	$h^{(1)},\ldots,h^{(m)}\in \Dbase$ satisfying
	\begin{align}
		\dim_\B\!\left(
			\operatorname{span}_\B
			\{h_0^{(1)},\ldots,h_0^{(m)}\}
		\right)
		&=m,
		\label{eq:gw:basic:t}\\
		\dim_\B\!\left(
			\operatorname{span}_\B
			\{h_j^{(1)},\ldots,h_j^{(m)}\}
		\right)
		&\leq 1,
		\qquad \forall\,j\in[\nbase],
		\label{eq:gw:basic:1}
	\end{align}
	then there exists a LERS requiring at most $\nbase$ $\B$-symbols.
\end{theorem}
\begin{proof}
	Let $\zeta_1,\ldots,\zeta_m$ be a basis of $\F$ over $\B$. By
	condition~\eqref{eq:gw:basic:t}, there exist dual codewords
	$h^{(1)},\ldots,h^{(m)}\in\Dbase$ such that $h_0^{(i)}=\zeta_i$ for $i\in[m].$
	By condition~\eqref{eq:gw:basic:1}, there exist 
	$\beta_1,\ldots,\beta_{\nbase}\in\F$ and coefficients
	$\lambda_j^{(i)}\in\B$ such that
	$h_j^{(i)}=\lambda_j^{(i)}\beta_j$ for each 
	$i\in[m]$ and $j\in[\nbase].$
	% Each $h^{(i)}$ can be written as
	% $h^{(i)}
	% =
	% (\zeta_i,
	% \lambda_1^{(i)}\cdot\beta_1,\ldots,
	% \lambda_{\nbase}^{(i)}\cdot \beta_{\nbase}).$
	For every
	$c=(c_0,\ldots,c_{\nbase})\in\Cbase$ and $i\in[m]$ we have
	\begin{align*}
	0
	=
	c_0\cdot h_0^{(i)}
	+\sum_{j\in[\nbase]}c_j \cdot h_j^{(i)}
	=
	c_0\cdot \zeta_i
	+\sum_{j\in[\nbase]}
	c_j\cdot \lambda_j^{(i)}\cdot \beta_j .
	\end{align*}
	Therefore,
	\[
	-\Tr_{\F/\B}(c_0 \cdot \zeta_i)
	=
	\sum_{j\in[\nbase]}
	\lambda_j^{(i)} \cdot
	\Tr_{\F/\B}(c_j\cdot \beta_j).
	\]
	Thus, from a single query $\Tr_{\F/\B}(c_j\cdot\beta_j)\in\B$
        for each $c_j$, we can compute 
        $\Tr_{\F/\B}(c_0 \cdot \zeta_i)$ for $i\in[m].$ Since
        $\{\zeta_1,\ldots,\zeta_m\}$ is a basis of $\F$ over $\B$, the maps $x\mapsto \Tr_{\F/\B}(c_0 \cdot x)$ are linearly independent, 
		so the  values $\Tr_{\F/\B}(c_0 \cdot \zeta_i)$ 
        determine $c_0$. \qed%Hence
                                                                            %the
                                                                            %erased
                                                                            %symbol
                                                                            %can
                                                                            % be reconstructed using at most $\nbase$ subsymbols. \qed
\end{proof}

\subsection{Linear Exact Repair Scheme via Subfield Subcodes}

% The dual characterization of LERS given by Guruswami and Wootter reduces
% the construction of a repair scheme to finding suitable dual codewords.
We recall the analysis of~\cite{wcc-paper}. First, punctured dual codewords
$\underline{h}^{(1)},\ldots,\underline{h}^{(m)}\in\underline{\Dbase}$
satisfying condition~\eqref{eq:gw:basic:1} are studied and then extended 
to the first coordinate so that they satisfy condition~\eqref{eq:gw:basic:t}. 
For this, let $\beta_1,\ldots,\beta_{\nbase}\in\F\setminus\set{0}$ be given, 
and consider punctured dual codewords whose $j$-th coordinates belong to the
line generated by $\beta_j$:
\[
\underline{\Dbase}_\beta
:=
\left\{
\underline{h}=(h_1,\ldots,h_{\nbase})\in\underline{\Dbase}
\;\middle|\;
h_j\in\operatorname{span}_{\B}(\beta_j),
\ \forall j\in[\nbase]
\right\}.
\]
Note that $\underline{\Dbase}_\beta$ is a $\B$-linear code over the
alphabet $\F$. For every $\underline{h}\in\underline{\Dbase}_\beta$,
there exist $\lambda_1,\ldots,\lambda_{\nbase}\in\B$ such that
$h_j=\lambda_j\beta_j$ for $j\in[\nbase]$. Therefore, define
\[
\Dlambdabeta:=
\left\{
(\lambda_1,\ldots,\lambda_{\nbase})\in\B^{\nbase}
\;\middle|\;
\exists\,\underline h\in\underline{\Dbase}_\beta
\text{ with }
h_j=\lambda_j\beta_j,\ \forall j\in[\nbase]
\right\},
\]
which is a $\B$-linear code over $\B$. More precisely, $
\Dlambdabeta
=
\underline{\Dbase}_\beta M_{\beta^{-1}}
\cap\B^{\nbase}$,
where $M_{\beta^{-1}}=\operatorname{Diag}(\beta_1^{-1},\ldots,\beta_{\nbase}^{-1}).$
Consider the extension map $\phi_{\Dbase}$ defined in~\eqref{eq:phi:D}.
Condition~\eqref{eq:gw:basic:t} requires the existence of
$m$ punctured dual codewords whose $\phi_{\Dbase}$-extensions
are $\B$-linearly independent. Consequently, a necessary condition is
$\dim_\B(\Dlambdabeta)\geq m.$
Using the dimension bound for subfield subcodes, we obtain
$\dim_\B(\Dlambdabeta)\geq \nbase-m\kbase.$
This bound suggests that the condition $
\kbase\leq \nbase/m-1
$
might be sufficient to guarantee the existence of a LERS. 
However, this dimension bound alone does not guarantee 
that the images of the corresponding 
$\phi_{\Dbase}$-extensions have dimension $m$.

\subsubsection{Necessary and Sufficient Condition.}
Consider the $\B$-linear map
	\begin{equation}\label{eq:phi:Lambda1}
			\psi :\begin{array}[t]{rcl}
				\Dlambdabeta&\longrightarrow &\F \\
				(\lambda_1,\ldots,\lambda_{\nbase})
				&\longmapsto&
				\sum_{i\in[\nbase]}\phi_i\beta_i\lambda_i 
			\end{array}.
    \end{equation}
Then 
$\beta_1,\ldots,\beta_{\nbase}$ will give a LERS if $\dim_\B(\operatorname{Im}(\psi))=m$.
In the following, a matrix characterization of this condition is derived.
%by explicitly deriving a parity-check matrix of $\Dlambdabeta$.
A parity-check matrix of $\underline{\Dbase}$ is given by
$\overline{\Gbase}$ from~\eqref{eq:Gbase}. Moreover,
\[
\begin{bmatrix}
\phi\\[.1em]
\overline{\Gbase}
\end{bmatrix}
\]
is a parity-check matrix of
$\underline{\Dbase}\cap\ker(\phi_{\Dbase}).$
For $\underline h\in\underline{\Dbase}_\beta$, let
$h=\lambda M_\beta$, where $M_\beta=\operatorname{Diag}(\beta_1,\ldots,\beta_{\nbase})$
and $\lambda\in\Dlambdabeta.$ Since $\underline h\in\underline{\Dbase}$, we have
$0=\overline{\Gbase}\underline h^T
=
\overline{\Gbase}(\lambda M_\beta)^T
=
(\overline{\Gbase}M_\beta)\lambda^T.
$
Consequently, a (not necessarily full-rank) parity-check matrix of
$\Dlambdabeta$ is given by
\[
(\overline{\Gbase}M_\beta)_\B
\in
\B^{m(\kbase-1)\times \nbase}.
\]
Similarly, $\ker(\psi)$ is the right kernel of the matrix
\begin{equation}
\label{eq:Kbeta}
K_\beta=
\begin{bmatrix}
(\phi M_\beta)_\B\\[.1em]
(\overline{\Gbase}M_\beta)_\B
\end{bmatrix}.
\end{equation}
Hence, the condition
$\dim_\B(\operatorname{Im}(\psi))=m$ is equivalent to requiring that
the rank of $K_\beta$ is exactly $m$ larger than the rank of a
parity-check matrix of $\Dlambdabeta,$ i.e.,
\begin{equation}\label{eq:CNS:t}
\rank_\B(K_\beta)
=
\rank_\B\big((\overline{\Gbase}M_\beta)_\B\big)+m .
\end{equation}
Therefore, any choice of
$\beta_1,\ldots,\beta_{\nbase}$ satisfying~\eqref{eq:CNS:t}
gives a linear exact repair scheme.

\subsubsection{Randomized Method.}
The rank characterization reduces the construction of a LERS to
finding $\beta_1,\ldots,\beta_{\nbase}\in\F\setminus \{0\}$
satisfying~\eqref{eq:CNS:t}. However, analyzing this rank is
difficult, since it requires controlling the dimension of 
subfield subcodes, which is not known in general. 
For
$\kbase\leq \nbase/m-1$, condition~\eqref{eq:CNS:t} is satisfied
with overwhelming probability when
$\beta_1,\ldots,\beta_{\nbase}$ are sampled independently and
uniformly from $\F\setminus\{0\}$. For the parameter choice
$q=2$ and $m=5$, and for a variety of parameter pairs $(n,k)$
(see~\cite{wcc-paper}), the condition was satisfied in all
$10\,000$ trials.

        \section{Analysis of Massey  Scheme for the Addition}\label{ch:add_gate}

         \subsection{Setup for the Addition}
		 
		Let $\Cbase$ be a code with generator matrix $\Gbase$ as in~\eqref{eq:Gbase},
		and let $\Hbase$ be a parity-check matrix of $\Cbase$. Consider three
		secret values $u,v,w\in\F$, which are masked by codewords
		$c_u,c_v,c_w\in\Cbase$, respectively, such that
		$c_{u,0}=u$, $c_{v,0}=v$, and $c_{w,0}=w$.	
		We focus on the addition operation $w=u+v$. The matrix
                		\[
           \Gobs=\begin{bmatrix}
             1 & 0 & 1\\
             0 & 1 & 1\end{bmatrix}
        \]
 		captures the relation between the 
		inputs $u,v$ and the output $w,$ with $
           (u,v,w)=(u,v) \cdot \Gobs.
         $
		The code $\Cobs$ generated by $\Gobs$ is referred to as the \textit{computational code}.
		Equivalently, it can be described by the
		parity-check matrix
 		\[
           \Hobs = \begin{bmatrix}
             1 & 1& -1\\\end{bmatrix},
        \]
         which describes the addition.
		Since the addition is a linear operation, the associated codewords (maskings) 
		satisfy $c_w=c_u+c_v$. Therefore, 
		\[
           (c_u,c_v,c_w)=(c_u,c_v) \begin{bmatrix}
             \textbf{1}_{\nbase+1} & \textbf{0}_{\nbase+1} & \textbf{1}_{\nbase+1}\\
             \textbf{0}_{\nbase+1} & \textbf{1}_{\nbase+1} & \textbf{1}_{\nbase+1}\end{bmatrix},
         \]
         where $\textbf{1}_{\nbase+1}$ and $\textbf{0}_{\nbase+1}$ denote the identity and zero matrices.
         % of 
	 %        size $(\nbase+1)\times(\nbase+1)$, respectively.
		The induced code
		\[
           \Cadd=\set{(c_u,c_v,c_w)\in \Cbase^3; c_w=c_u+c_v},
        \]
		is a $[3(\nbase+1),2\kbase]_{\F}$ code with generator matrix
        \[
           \Gadd=\begin{bmatrix}
             \Gbase &      \textbf{0} & \Gbase\\
             \textbf{0}      & \Gbase & \Gbase
           \end{bmatrix}.
         \]
         Actually $\Gadd$ can be expressed as the Kronecker product
         $\Gadd=\Gobs\otimes G_0,$ and $\Cadd=\Cobs\otimes \Cbase$ is
         the \emph{product code} induced by $\Cobs$ and $\Cbase$. The
         dual code of $\Cadd$, denoted by $\Dadd$, is a
         $[3(\nbase+1),3(\nbase+1)-2\kbase]_{\F}$ code. A redundant
         parity-check matrix for $\Cadd$ (equivalently, a generator matrix of
         $\Dadd$) is given by
         \[
           \begin{bmatrix}
             \Hbase &      \textbf{0} &      \textbf{0} \\
             \textbf{0}      & \Hbase &      \textbf{0} \\
             \textbf{0}      &      \textbf{0} & \Hbase \\
             \textbf{1}_{\nbase+1}    &    \textbf{1}_{\nbase+1} &    -\textbf{1}_{\nbase+1} \\
           \end{bmatrix}.
         \]
		 By removing redundant rows, we obtain a full-rank parity-check matrix
         \begin{equation}
           \Hadd =
           \begin{bmatrix}
             \Hbase &      \textbf{0} &      \textbf{0} \\
             \textbf{0}      & \Hbase &      \textbf{0} \\
             \textbf{1}_{\nbase+1}    &    \textbf{1}_{\nbase+1} &    -\textbf{1}_{\nbase+1} \\
           \end{bmatrix},\label{eq:pcheck:D}
         \end{equation}
		 whose rank is exactly $3(\nbase+1)-2\kbase.$ 
		 We write
		$c_u=(u,\underline{c_u})$ for $\underline{c_u}\in\F^{\nbase},$
		and analogously for $c_v$ and $c_w$. Similarly, a dual codeword
		$h\in\Dadd$ is denoted by
		\[
			h=(h_u,h_v,h_w)
			=
			\bigl(
				(h_{u,0},\underline{h_u}),
				(h_{v,0},\underline{h_v}),
				(h_{w,0},\underline{h_w})
			\bigr).
		\]
		Thus, for
			$c=
			\bigl(
				(u,\underline{c_u}),
				(v,\underline{c_v}),
				(w,\underline{c_w})
			\bigr)
			\in\Cadd$
		and $h\in\Dadd$, we have
		\begin{equation}
		\label{eq:scalar-product}
			h_{u,0}\cdot u
			+\underline{h_u}\cdot\underline{c_u}^{\,T}
			+h_{v,0}\cdot v
			+\underline{h_v}\cdot\underline{c_v}^{\,T}
			+h_{w,0}\cdot w
			+\underline{h_w}\cdot\underline{c_w}^{\,T}
			=0.
		\end{equation}
		We abuse notation and write
		$\underline{h}=(\underline{h_u},\underline{h_v},\underline{h_w}),$
		where
		$\underline{h_u}=(h_{u,1},\ldots,h_{u,\nbase}),$ 
		$\underline{h_v}=(h_{v,1},\ldots,h_{v,\nbase}),$ and
		$\underline{h_w}=(h_{w,1},\ldots,h_{w,\nbase}).$
		Define the coordinate sets
		$I_0=\{0,\nbase+1,2\nbase+2\}$ and
			$J=\{0,\ldots,3(\nbase+1)-1\}\setminus I_0.$
		We write
		$\underline{h}=(h_i)_{i\in J},$
		where the coordinates are indexed by the coordinate set $J$
		obtained by puncturing at $I_0$.
		More precisely, letting $J=J_u \sqcup J_v\sqcup J_w$ denote the corresponding partition of $J$, 
		we also use the notation
		\[
		\underline{h}=\big((h_i)_{i\in J_u},\,(h_i)_{i\in J_v},\,(h_i)_{i\in J_w}\big).
		\]
		% which is consistent with the block decomposition 
		% $\underline{h}=(\underline{h_u},\underline{h_v},\underline{h_w})$ introduced above.

		\subsection{Repair Problem for the Addition}

		In this setting, the repair problem consists of recovering $(u,v,w)$
		from $(\underline{c_u},\underline{c_v},\underline{c_w})$.
                % Since $w=u+v,$
		% any two of the three values $u$, $v$, and $w$ uniquely determine the
		% remaining one.
		The linear exact repair problem can be reformulated in terms of parity
		equations. We will show that the missing symbols can be recovered if
		there exists a matrix $H_{2m}$ consisting of $2m$ codewords
		$h^{(1)},\ldots,h^{(2m)}\in\Dadd$ of the form
		\[
			H_{2m}
			=
			\begin{bmatrix}
				h^{(1)}\\
				h^{(2)}\\
				\vdots\\
				h^{(2m)}
			\end{bmatrix}
			=
			\begin{bmatrix}
				h_{u,0}^{(1)} & \cdots & h_{u,\nbase}^{(1)}
				& h_{v,0}^{(1)} & \cdots & h_{v,\nbase}^{(1)}
				& h_{w,0}^{(1)} & \cdots & h_{w,\nbase}^{(1)}\\
				h_{u,0}^{(2)} & \cdots & h_{u,\nbase}^{(2)}
				& h_{v,0}^{(2)} & \cdots & h_{v,\nbase}^{(2)}
				& h_{w,0}^{(2)} & \cdots & h_{w,\nbase}^{(2)}\\
				\vdots & \ddots & \vdots
				& \vdots & \ddots & \vdots
				& \vdots & \ddots & \vdots\\
				h_{u,0}^{(2m)} & \cdots & h_{u,\nbase}^{(2m)}
				& h_{v,0}^{(2m)} & \cdots & h_{v,\nbase}^{(2m)}
				& h_{w,0}^{(2m)} & \cdots & h_{w,\nbase}^{(2m)}
			\end{bmatrix},
		\]
		such that the matrix 
		\begin{equation}
			A=
			\begin{bmatrix}
				h_{u,0}^{(1)}+h_{w,0}^{(1)}
				&
				h_{v,0}^{(1)}+h_{w,0}^{(1)}
				\\
				\vdots & \vdots
				\\
				h_{u,0}^{(2m)}+h_{w,0}^{(2m)}
				&
				h_{v,0}^{(2m)}+h_{w,0}^{(2m)}
			\end{bmatrix}\label{matrix:A:add}
		\end{equation}
               satisfies 		\begin{equation}
			\rank_{\B}\!\left(A^{\B}\right)=2m,
			\label{eq:lin_indep_mpc}
		\end{equation}
		where 
		$A^\B\in\B^{2m\times 2m}$ is obtained by expanding each entry of $A$
		as a row vector of its coordinates in $\B$.
		In addition, for each $j\in[\nbase]$, we require 
		\begin{align}
			\dim_{\B}\!\left(
				\operatorname{span}_{\B}
				\left\{
					h_{u,j}^{(1)},\ldots,h_{u,j}^{(2m)}
				\right\}
			\right)
			&= 1,
			\label{eq:linear_dep_const_mpc:u}
			\\
			\dim_{\B}\!\left(
				\operatorname{span}_{\B}
				\left\{
					h_{v,j}^{(1)},\ldots,h_{v,j}^{(2m)}
				\right\}
			\right)
			&= 1,
			\label{eq:linear_dep_const_mpc:v}
			\\
			\dim_{\B}\!\left(
				\operatorname{span}_{\B}
				\left\{
					h_{w,j}^{(1)},\ldots,h_{w,j}^{(2m)}
				\right\}
			\right)
			&= 1.
			\label{eq:linear_dep_const_mpc:w}
		\end{align}
         
		\begin{theorem}\label{thm:existence_add_gate}
			Let $\Cadd$ be the $[3(\nbase+1),2\kbase]_{\F}$ product code induced by
		    $\Cbase$ and $\Cobs$ as described above. If there exists a matrix $H_{2m}$
			satisfying \eqref{eq:lin_indep_mpc} and
			\eqref{eq:linear_dep_const_mpc:u}--\eqref{eq:linear_dep_const_mpc:w},
			then there exists an equally distributed LERS for $\Cadd$.
		\end{theorem}

		\begin{proof}
			% For all $c=(c_u,c_v,c_w)\in\Cadd$ and
			% $h=(h_u,h_v,h_w)\in\Dadd$, Equation~\eqref{eq:scalar-product}
			% holds.
                        Let $H_{2m}$ be a matrix satisfying
			\eqref{eq:lin_indep_mpc}. Since $\rank_{\B}(A^\B)=2m$, elementary
			$\B$-linear row operations transform $A^\B$ into
			\[
				\begin{bmatrix}
					\textbf{1}_m & \textbf{0}\\
					\textbf{0} & \textbf{1}_m
				\end{bmatrix},
			\]
			where $\textbf{1}_m$ denotes the $m\times m$ identity matrix. Thus, there exist dual codewords
			$h^{(1)},\ldots,h^{(2m)}\in\Dadd$ such that
			\[
				\begin{array}{ll}
					h_{u,0}^{(i)}+h_{w,0}^{(i)}=\zeta_i,
						& 1\leq i\leq m,\\
					h_{u,0}^{(i)}+h_{w,0}^{(i)}=0,
						& m+1\leq i\leq 2m,\\
			% 	\end{cases}
			% \]
			% and
			% \[
			% 	\begin{cases}
					h_{v,0}^{(i)}+h_{w,0}^{(i)}=0,
						& 1\leq i\leq m,\\
					h_{v,0}^{(i)}+h_{w,0}^{(i)}=\zeta_i,
						& m+1\leq i\leq 2m.
				\end{array}
			\]
			Furthermore, for each $i\in[2m]$, conditions
			\eqref{eq:linear_dep_const_mpc:u}--\eqref{eq:linear_dep_const_mpc:w}
			imply that, for every $j\in[\nbase]$, there exist
			$\lambda_{u,j}^{(i)},\lambda_{v,j}^{(i)},\lambda_{w,j}^{(i)}\in\B$
			and $\beta_{u,j},\beta_{v,j},\beta_{w,j}\in\F$ such that
			  $h_{u,j}^{(i)} = \lambda_{u,j}^{(i)}\beta_{u,j},\text{ }
				h_{v,j}^{(i)} = \lambda_{v,j}^{(i)}\beta_{v,j}, \text{ and }				
				h_{w,j}^{(i)} = \lambda_{w,j}^{(i)}\beta_{w,j}.$
			For each $i\in[m]$, Eq~\eqref{eq:scalar-product} gives
			\begin{align*}
				-c_{u,0}h_{u,0}^{(i)}
				-c_{v,0}h_{v,0}^{(i)}
				-c_{w,0}h_{w,0}^{(i)}
				&=
				\sum_{j=1}^{\nbase}
				\Big(
					c_{u,j}h_{u,j}^{(i)}
					+c_{v,j}h_{v,j}^{(i)}
					+c_{w,j}h_{w,j}^{(i)}
				\Big)\\
				-c_{u,0}
				\big(
					h_{u,0}^{(i)}+h_{w,0}^{(i)}
				\big)
				-c_{v,0}
				\big(
					h_{v,0}^{(i)}+h_{w,0}^{(i)}
				\big)
				&=
				\sum_{j=1}^{\nbase}
				\Big(
					c_{u,j}\lambda_{u,j}^{(i)}\beta_{u,j}
					+c_{v,j}\lambda_{v,j}^{(i)}\beta_{v,j}\\
				&\qquad
					+c_{w,j}\lambda_{w,j}^{(i)}\beta_{w,j}
				\Big).
			\end{align*}
			Hence, 
			\begin{align*}
				-c_{u,0}\zeta_i
				=
				\sum_{j=1}^{\nbase}
				\Big(
					c_{u,j}\lambda_{u,j}^{(i)}\beta_{u,j}
					+c_{v,j}\lambda_{v,j}^{(i)}\beta_{v,j}
					+c_{w,j}\lambda_{w,j}^{(i)}\beta_{w,j}
				\Big).
			\end{align*}
			Applying the trace operator $\Tr_{\F/\B}$ to the preceding equation, we obtain
			\[
			\begin{aligned}
				-\Tr_{\F/\B}\!\left(c_{u,0}\zeta_i\right)
				={}&
				\sum_{j=1}^{\nbase}
				\lambda_{u,j}^{(i)}
				\Tr_{\F/\B}\!\left(c_{u,j}\beta_{u,j}\right)\\
				&+
				\sum_{j=1}^{\nbase}
				\lambda_{v,j}^{(i)}
				\Tr_{\F/\B}\!\left(c_{v,j}\beta_{v,j}\right)
				+
				\sum_{j=1}^{\nbase}
				\lambda_{w,j}^{(i)}
				\Tr_{\F/\B}\!\left(c_{w,j}\beta_{w,j}\right).
			\end{aligned}
			\]
			Thus, after querying the $3\nbase$ subsymbols
			\[
				\Tr_{\F/\B}\!\left(c_{u,j}\beta_{u,j}\right),\qquad
				\Tr_{\F/\B}\!\left(c_{v,j}\beta_{v,j}\right),\qquad
				\Tr_{\F/\B}\!\left(c_{w,j}\beta_{w,j}\right),
				\qquad j\in[\nbase],
			\]
			we can, by varying $i\in[m]$ (and hence $\zeta_i$), compute the $m$
			linear functions
				$\Tr_{\F/\B}\!\left(c_{u,0}\zeta_i\right)$
				for $i\in[m].$
			Indeed, each of these values is determined by the queried subsymbols
			and the fixed, known coefficients
			$\lambda_{u,j}^{(i)},\lambda_{v,j}^{(i)},\lambda_{w,j}^{(i)}\in\B$.
			Since $\{\zeta_1,\ldots,\zeta_m\}$ is a basis over $\B$ of $\F$, these
			$m$ linear functions  determine $c_{u,0}$.

			Similarly, using the remaining dual codewords $h^{(m+1)},\ldots,h^{(2m)}\in\Dadd$, for $i\in[m]$, we can determine the  $m$ trace values
			$\Tr_{\F/\B}(c_{v,0}\zeta_i)$ using the analog queries as above, which
			uniquely determine $c_{v,0}$. \qed
		\end{proof}
		\subsection{Subfield Subcode Analysis}
		% In the previous subsection, we saw that, if there exists a matrix $H_{2m}$ satisfying the full-rank condition
		% \eqref{eq:lin_indep_mpc} and the one-dimensional column-span conditions
		% \eqref{eq:linear_dep_const_mpc:u}--\eqref{eq:linear_dep_const_mpc:w},
		% then we can construct an equally distributed LERS for $\Cadd.$ 
% 		We focus on the construction of a matrix $H_{2m}$ satisfying the full-rank condition
% 		\eqref{eq:lin_indep_mpc} and the one-dimensional  conditions
% 		\eqref{eq:linear_dep_const_mpc:u}--\eqref{eq:linear_dep_const_mpc:w}. 
% %		\subsubsection{Fullfilling the One-Dimensional Column-Span Conditions.}
		First, we construct a matrix $H_{2m}$ satisfying 
%		conditions
		\eqref{eq:linear_dep_const_mpc:u}--\eqref{eq:linear_dep_const_mpc:w}. Let $\underline{\Dadd}$ denote the puncturing of the dual code $\Dadd$
		at the coordinates in $I_0$, with coordinate set $J$. 
		% \begin{equation}
		% 	\Gbase=
		% 	\begin{bmatrix}
		% 		-1 & \phi\\
		% 		\textbf{0} & \overline{\Gbase}
		% 	\end{bmatrix},
		% 	\label{eq:Gbase}
		% \end{equation}
		 The following matrix is a generator matrix of $\Cadd=\Cobs\otimes\Cbase$  
		 and, consequently, a parity-check matrix of $\Dadd$:
		\[
			\Gadd=
			\begin{bmatrix}
				-1 & \phi & 0 & \textbf{0} & -1 & \phi\\
				\textbf{0} & \overline{\Gbase} & \textbf{0} & \textbf{0} & \textbf{0} & \overline{\Gbase}\\
				0 & \textbf{0} & -1 & \phi & -1 & \phi\\
				\textbf{0} & \textbf{0} & \textbf{0} & \overline{\Gbase} & \textbf{0} & \overline{\Gbase}
			\end{bmatrix}.
		\]
		Moreover,
		\begin{equation}
			\overline{\Gadd}
			=
			\begin{bmatrix}
				\overline{\Gbase} & \textbf{0} & \overline{\Gbase}\\
				\textbf{0} & \overline{\Gbase} & \overline{\Gbase}
			\end{bmatrix},
			\label{eq:parity:under:Dadd}
		\end{equation}
		is a generator matrix for $\overline{\Cadd}$, where
		$\overline{\Cadd}$ is the dual code of the punctured code
		$\underline{\Dadd}$.
		Puncturing the parity-check matrix of $\Cadd$ at 
		$I_0$ yields
		\[
			\begin{bmatrix}
				\underline{\Hbase} & \textbf{0} & \textbf{0}\\
				\textbf{0} & \underline{\Hbase} & \textbf{0}\\
				\textbf{0} & \textbf{0} & \textbf{0}\\
				\textbf{1}_{\nbase} & \textbf{1}_{\nbase} & -\textbf{1}_{\nbase}
			\end{bmatrix}
		\]
		and the dimension of the punctured code $\underline{\Dadd}$
		is $
			\dim_\B(\underline{\Dadd})
			=
			3(\nbase+1)-2\kbase-1
		$.
		For each $j\in J$, let $\beta_j\in\F\setminus\{0\}$ be such that the
		$j$-th column of $H_{2m}$ is contained in
		$\operatorname{span}_{\B}(\beta_j)$.
		Let $\Daddbeta$ be defined by
		\[
			\Daddbeta
			:=
			\left\{
				\underline{h}\in\underline{\Dadd}
				\;\middle|\;
				h_j\in\operatorname{span}_{\B}(\beta_j)
				\text{ for all }j\in J
			\right\},
		\]
		which is a $\B$-linear code of length $3\nbase$ over the alphabet $\F$.
        We define
		\[
			\Dlambdabeta
			:=
			\left\{
				\underline{\lambda}\in\B^{3\nbase}
				\;\middle|\;
				\exists\,\underline{h}\in\Daddbeta
				\text{ such that }
				h_j=\lambda_j\beta_j
				\text{ for all }j\in J
			\right\},
		\]
		which is a $\B$-linear code of length $3\nbase$ over
		the alphabet $\B$.
		With the diagonal matrix
			$M_{\beta^{-1}}
			:=
			\operatorname{Diag}\!\left(\beta_j^{-1}\right)_{j\in J}
			\in\F^{3\nbase\times 3\nbase}$, we have, 
			similarly to the base case, that
			\[
				\Dlambdabeta
				=
				\left(
					\underline{\Dadd}\cdot M_{\beta^{-1}}
				\right)
				\cap \B^{3\nbase},
                              \]
                             is a subfield subcode.
The redundancy of $\Daddbeta$ is $
			r(\Daddbeta)
			=
			n(\Daddbeta)-k(\Daddbeta)
			=
			2\kbase-2$.
		Using the lower bound on the dimension 
			$\dim_{\B}(\Dlambdabeta)
			\geq
			n(\Daddbeta)-m\cdot r(\Daddbeta)$, we obtain
		\begin{equation*}
			\dim_{\B}(\Dlambdabeta)
			\geq
			3\nbase-2m(\kbase-1).
%			\label{eq:dim:sub:gen:case}
		\end{equation*}
		To ensure that $\dim_{\B}(\Dlambdabeta)\geq 2m,$ it suffices to have
			$$\kbase \leq \frac{3\nbase}{2m}.$$
		Hence, there exist $2m$ linearly independent elements of $\Dlambdabeta$ 
		satisfying the conditions
		\eqref{eq:linear_dep_const_mpc:u}--\eqref{eq:linear_dep_const_mpc:w}.
%		\subsubsection{Fullfilling the Full-Rank Condition.}
		Once we have
		$\underline{h}=(\underline{h_u},\underline{h_v},\underline{h_w})\in\Daddbeta,$
		we seek to extend $\underline{h}$ to codewords
		$h_u,h_v,h_w$ satisfying the full-rank condition
		\eqref{eq:lin_indep_mpc}. From~\ref{matrix:A:add}, we are concerned with the values 
		$h_{u,0}+h_{w,0}$
		and
		$h_{v,0}+h_{w,0}$.
		\begin{definition}[Add extension map]\label{def:add:ext:map}
			Let $\Gbase$ be a generator matrix as in \eqref{eq:Gbase}, whose
			first row is $(-1,\phi)$. We define the \emph{add extension map}
			$\phi_{\Dadd}$ by
			\begin{equation}
				\phi_{\Dadd} \colon \begin{array}[t]{rcl}
					\underline{\Dadd}& \longrightarrow &\F^2 \\
					(\underline{h_u},\underline{h_v},\underline{h_w})& \longmapsto &
					\left(\phi\cdot \underline h_u^T+\phi\cdot \underline h_w^T,\ \phi\cdot \underline h_v^T+\phi\cdot \underline h_w^T\right).
				\end{array}
			\end{equation}
		\end{definition}
		Applying this map to
		$\underline{h}\in\Daddbeta$, and writing
			$h_{u,j}=\lambda_{u,j}\beta_{u,j}$, 
			$h_{v,j}=\lambda_{v,j}\beta_{v,j}$, 
			$h_{w,j}=\lambda_{w,j}\beta_{w,j}$ for $ j\in[\nbase]$
		with $\underline{\lambda}\in\Dlambdabeta$,
		gives the $\B$-linear map
		\begin{equation}\label{eq:phi:Lambda111}
			\psi_{\Dlambdabeta} : \begin{array}[t]{rcl}
				\Dlambdabeta& \rightarrow &\F^2 \\
				\underline{\lambda} & \mapsto &
				(
					\sum_{j=1}^{\nbase}
					\phi_j
					\bigl(
						\lambda_{u,j}\beta_{u,j}
						+
						\lambda_{w,j}\beta_{w,j}
					\bigr),
					\sum_{j=1}^{\nbase}
					\phi_j
					\bigl(
						\lambda_{v,j}\beta_{v,j}
						+
						\lambda_{w,j}\beta_{w,j}
					\bigr)
				).
			\end{array}
		\end{equation}
		It remains to ensure that
		\begin{equation}\label{eq:dim:3m}
		 	\dim_{\B}\bigl(\operatorname{Im}(\psi_{\Dlambdabeta})\bigr)=2m.
		\end{equation}
		% then there exist codewords
		% $\underline{h}^{(1)},\ldots,\underline{h}^{(2m)}\in\Daddbeta$
		% whose images under $\psi_{\Dlambdabeta}$ form a basis over $\B$ of $\F^2$ of the form
		% $(\zeta_1,0),\ldots,(\zeta_m,0),(0,\zeta_1),\ldots,(0,\zeta_m).$
		% Let $\ker(\psi_{\Dlambdabeta})$ and
		% $\operatorname{Im}(\psi_{\Dlambdabeta})$ denote the kernel and image of
		% $\psi_{\Dlambdabeta}$, respectively. By the rank--nullity theorem,
		% 	$\dim_{\B}\ker(\psi_{\Dlambdabeta})
		% 	+
		% 	\dim_{\B}\operatorname{Im}(\psi_{\Dlambdabeta})
		% 	=
		% 	\dim_{\B}\Dlambdabeta.$
		% Consequently, if $\dim_{\B}\operatorname{Im}(\psi_{\Dlambdabeta})=2m,$
		% then
		% 	$\dim_{\B}\ker(\psi_{\Dlambdabeta})
		% 	=
		% 	\dim_{\B}\Dlambdabeta-2m.$
		% In particular, we have
		% 	$\dim_{\B}\Dlambdabeta\geq 2m.$
		% With the bound \eqref{eq:dim:sub:gen:case}, this yields the
		% condition
		% \begin{equation}
		% 	\kbase\leq\frac{3\nbase}{2m}.
		% 	\label{eq:cond1}
		% \end{equation}
		Previously, we have shown that $\kbase\leq3\nbase /(2m)$ is sufficient to ensure that
		$\dim_{\B}(\Dlambdabeta)\geq 2m$. However, this does not imply that
		$\dim_{\B}\bigl(\operatorname{Im}(\psi_{\Dlambdabeta})\bigr)\geq 2m.$
		% Determining the dimension of
		% $\operatorname{Im}(\psi_{\Dlambdabeta})$ is equivalent to computing the
		% dimension of a subfield subcode, which is known to be difficult in
		% general. 
                % We therefore adopt a randomized approach to find suitable
		% values $\beta_j\in\F$ for $j\in J$,
		% such that
		% 	$\dim_{\B}\bigl(\operatorname{Im}(\psi_{\Dlambdabeta})\bigr)=2m.$
%		For this, we first introduce parity-check matrices that are used later to analyse this condition.
        Let us first make the condition in~\eqref{eq:dim:3m} explicit.
		Recall that the matrix in
		\eqref{eq:parity:under:Dadd} is a parity-check matrix for
		$\underline{\Dadd}$. 
		Define $M_{3\beta}
			=
			\operatorname{Diag}
			\bigl(
				M_{u,\beta},
				M_{v,\beta},
				M_{w,\beta}
			\bigr)$,
		where
		\[
			M_{u,\beta}
			=
			\operatorname{Diag}\bigl((\beta_j)_{j\in J_u}\bigr),
			\quad
			M_{v,\beta}
			=
			\operatorname{Diag}\bigl((\beta_j)_{j\in J_v}\bigr),
			\quad
			M_{w,\beta}
			=
			\operatorname{Diag}\bigl((\beta_j)_{j\in J_w}\bigr).
		\]
		Since
		\[
			0
			=
			\overline{\Gadd}\,\underline{h}^{\,T}
			=
			\overline{\Gadd}
			\left(\underline{\lambda}M_{3\beta}\right)^{T}
			=
			\left(\overline{\Gadd}M_{3\beta}\right)
			\underline{\lambda}^{\,T},
		\]
		a (not necessarily full-rank) parity-check matrix for
		$\Dlambdabeta$ is given, after expanding the columns over $\B$, by
		\begin{equation}\label{eq:H_Lambdabeta0}
			H_{\Dlambdabeta}
			=
			\begin{bmatrix}
				(\overline{\Gadd_0}M_{u,\beta})_{\B} & \mathbf{0} &
				(\overline{\Gadd_0}M_{w,\beta})_{\B}\\
				\mathbf{0} & (\overline{\Gadd_0}M_{v,\beta})_{\B} &
				(\overline{\Gadd_0}M_{w,\beta})_{\B}
			\end{bmatrix}
			\in
			\B^{m\cdot 2(\kbase-1)\times 3\nbase}
		\end{equation}
		and the matrix
		\begin{equation}\label{eq:H_Lambdabeta}
			H^0_{\Dlambdabeta}
			=
			\begin{bmatrix}
				(\phi M_{u,\beta})_{\B} & \mathbf{0} &
				(\phi M_{w,\beta})_{\B}\\
				(\overline{\Gbase}M_{u,\beta})_{\B} & \mathbf{0} &
				(\overline{\Gbase}M_{w,\beta})_{\B}\\
				\mathbf{0} & (\phi M_{v,\beta})_{\B} &
				(\phi M_{w,\beta})_{\B}\\
				\mathbf{0} & (\overline{\Gbase}M_{v,\beta})_{\B} &
				(\overline{\Gbase}M_{w,\beta})_{\B}
			\end{bmatrix}
		\end{equation}
		is a parity-check matrix for
		$\ker(\psi_{\Dlambdabeta})$.
                % , since the additional rows describe the
		% the add extension map \eqref{def:add:ext:map}.
		Consequently,
		\begin{equation}
			\label{eq:Dlambdabeta}
			\dim_{\B}(\Dlambdabeta)
			=
			3\nbase
			-
			\operatorname{rank}_{\B}\!\left(
				H_{\Dlambdabeta}
			\right).
		\end{equation}
		Furthermore,
                % since 
		% $\ker(\psi_{\Dlambdabeta})=\{x\in\F^{3\nbase}\colon H^0_{\Dlambdabeta} x^T=0\},$
		the rank-nullity theorem implies 
		\begin{equation}
			\label{eq:ker:Dlambdabeta}
			\dim_{\B}(\ker(\psi_{\Dlambdabeta}))
			=
			3\nbase
			-
			\operatorname{rank}_{\B}\!\left(
				H^0_{\Dlambdabeta}
			\right).
		\end{equation}
		By \eqref{eq:Dlambdabeta} and \eqref{eq:ker:Dlambdabeta}, the condition
		$\dim_{\B}\bigl(\operatorname{Im}(\psi_{\Dlambdabeta})\bigr)=2m$
		is equivalent to
		\begin{equation}
			\operatorname{rank}_{\B}\!\left(H^0_{\Dlambdabeta}\right)
			-
			\operatorname{rank}_{\B}\!\left(H_{\Dlambdabeta}\right)
			=2m.
			\label{eq:rkcond:Im:psi}
		\end{equation}
		Thus, any choice of $\beta_j\in\F\backslash\{0\}$, $j\in J$, satisfying
		\eqref{eq:rkcond:Im:psi} gives a linear exact repair scheme.
		Again, a theoretical analysis of condition~\eqref{eq:rkcond:Im:psi}
		is not known. We therefore resort to a random search.
		% To summarize, the subfield subcode construction succeeds
		% if the dimensional threshold (\ref{eq:cond1}) for the product code is fulfilled 
		% and suitable $\beta$-values are determined via the randomized method.
		
		\subsection{Simulations}
		\begin{figure}[ht]
			\begin{center}
				\includegraphics[width=\textwidth]{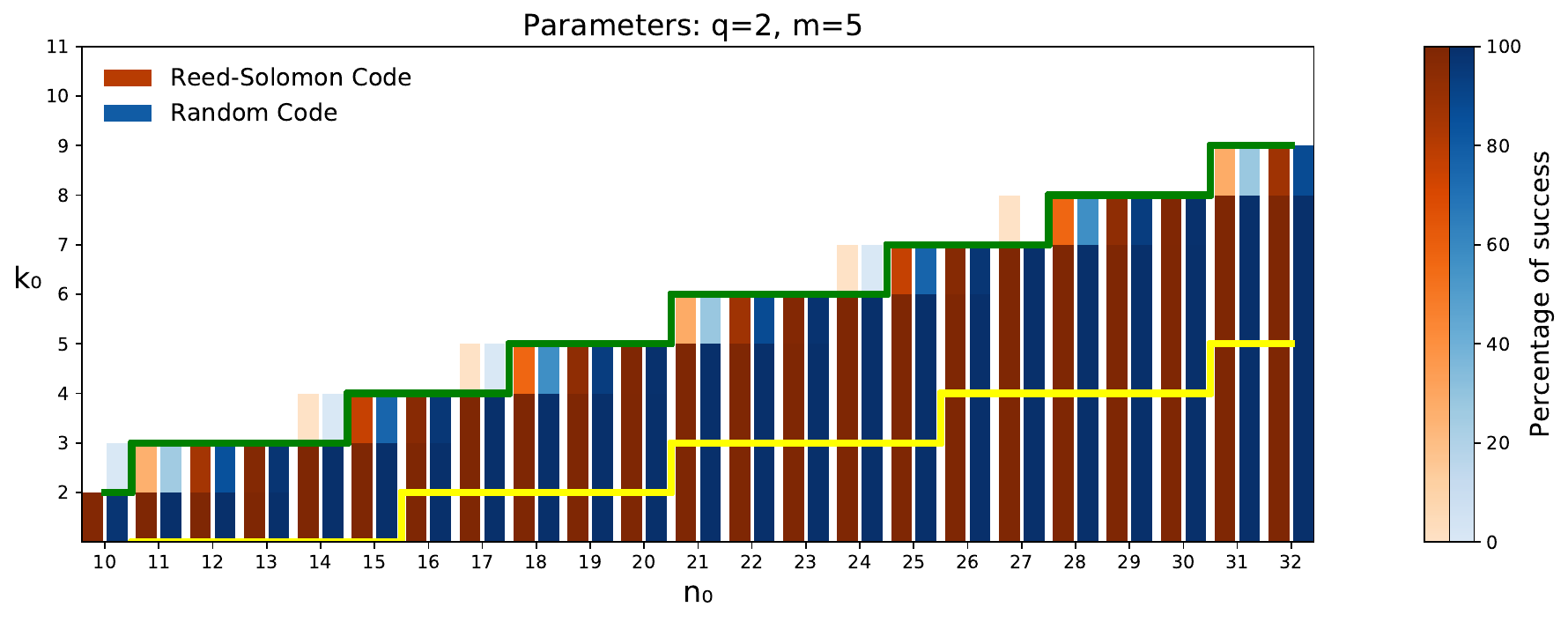}
			\end{center}
			\caption{Percentage of success for the addition (based on $10^6$ trials for each pair  $(\nbase,\kbase)$) for $q=2$ and 
			$m=5$. Here, success 
			means that \eqref{eq:rkcond:Im:psi} holds. The upper bound $\kbase\leq 3\nbase/2m$ is indicated in green, while the 
			bound  for the base case is indicated in yellow.}\label{successprob_q2_t5}
		\end{figure}
				We implemented the randomized method in MAGMA and
				performed simulations for various parameter
				choices. The results for fixed $q=2$ and $m=5$ are
				shown in Figure~\ref{successprob_q2_t5} for both
				Reed--Solomon codes and randomly generated codes.
				Following the randomized method, for each pair
				$(\nbase,\kbase)$, $10^6$ independent trials were
				performed. In each trial, the elements
				$\beta_{u,j},\beta_{v,j},\beta_{w,j}\in\F\setminus\{0\}$
				are chosen independently and uniformly at random. Then the matrices in
                \eqref{eq:H_Lambdabeta0} and \eqref{eq:H_Lambdabeta}
                are constructed, and it is verified whether the rank
                condition~\eqref{eq:rkcond:Im:psi} holds. As expected,
                this condition is satisfied with high probability when
                $\kbase\leq 3 \kbase/(2m)$ is met (highlighted in
                green). This demonstrates that the
                upper bound  for $\kbase$ allows leakage for more Massey secret sharing schemes.

                In other words, the secrets of the corresponding Massey secret
                sharing schemes for the addition  are vulnerable
                to leakage. Furthermore, note that the success
                probability of finding suitable $\beta$ values increases
                as $\kbase$ decreases.  For comparison, the yellow
                line in the figures represents the upper bound
                of the construction
                of~\cite{wcc-paper}. The gap between these two
                bounds illustrates the improvement achieved by our
                construction.

	%%%%%%%%%%%%%%%%%%%%%%%%%%%%%%%%%%%%%%%%%%%%%%%%%%
	%%%%%%%%%%%%%%%%%%%%%%%%%%%%%%%%%%%%%%%%%%%%%%%%%%
	% Generalization to any linear operational gate %
	%%%%%%%%%%%%%%%%%%%%%%%%%%%%%%%%%%%%%%%%%%%%%%%%%%
	%%%%%%%%%%%%%%%%%%%%%%%%%%%%%%%%%%%%%%%%%%%%%%%%%%
		
\section{Generalization to Linear Computations}

	\subsection{Masking of General Linear Computations}
		The computation proceeds as follows. 
		Given $\F$-linearly independent input values
		$u_1,\ldots,u_K\in\mathbb{F}$ and intermediate or output values
		$w_1,\ldots,w_{N-K}\in\mathbb{F},$ the computation of the  
		values $w_1,\ldots,w_{N-K}$ from the inputs $u_1,\ldots,u_K$ can be
		described by a matrix 
        $\Hobs \in \mathbb{F}^{(N-K)\times N}$ of the form
		\[
		\Hobs =
		\kbordermatrix{
		& u_1 & \cdots & u_K & w_1 & \cdots & w_i & \cdots & w_{N-K} \\
		g_1 & g_{1,1} & \cdots & g_{1,K} & -1 & \cdots & 0 & \cdots & 0 \\
		\vdots & \vdots & & \vdots & \vdots & & \vdots & & \vdots \\
		g_i & g_{i,1} & \cdots & g_{i,K} & g_{i,K+1} & \cdots & -1 & \cdots & 0 \\
		g_{i+1} & g_{i+1,1} & \cdots & g_{i+1,K} & g_{i+1,K+1} & \cdots & g_{i+1,K+i} & \cdots & 0 \\
		\vdots & \vdots & & \vdots & \vdots & & \vdots & & \vdots \\
		g_{N-K} & g_{N-K,1} & \cdots & g_{N-K,K} & g_{N-K,K+1} & \cdots & g_{N-K,K+i} & \cdots & -1
		}.
		\]
		Each row $g_i$ describes a step of a linear operation, that expresses the  value  $w_i$ as a linear combination 
		of the input values and the previously computed intermediate values.
		We define the associated $[N,K]_\F$ code by
		\[
		\Cobs=\{c\in\mathbb{F}^N:\Hobs \cdot c^T=0\},
		\]
		which we again call the \emph{computation code}.
		Let $\Gobs$ be the systematic generating matrix of
                $\Cobs,$ whose first $K$ columns form the identity matrix. Then
		\[
			(u_1,\dots,u_K,w_1,\dots,w_{N-K})= (u_1,\dots,u_K) \cdot \Gobs.
		\]
		Masking is done using the base code $\Cbase$,  such that each $u_i$ (respectively $w_i$) has an  associated  $c_{u_i}\in\Cbase$
		(respectively $c_{w_i}\in\Cbase$).
                % As in the previous section, the base code $\Cbase$ is generated by the systematic generator matrix
		%$\Gbase \in \mathbb{F}^{k_0 \times (n_0+1)}$.
		Then, each computation is applied in parallel to
		each coordinate of the $c_{u_i}$'s and the $c_{w_i}$'s i.e., at step $g_i,$ we have
		\[
			c_{w_i}= g_{i,1}c_{u_1}+\dots+g_{i,K}c_{u_K}+g_{i,K+1}c_{w_1}+\cdots+g_{i,K+i-1}c_{w_{i-1}}. 
		\]
		Then the codeword
		$c=(c_{u_1},\dots,c_{u_K},c_{w_1},\dots,c_{w_{N-K}})$
		belongs to the $[N(n_0+1), Kk_0]_{\F}$ product code   $\Clin= \Cobs \otimes \Cbase.$ 
		For simplicity, we index $c$ as follows
		\[
		c=(c_{u_1},\dots,c_{u_K},c_{w_1},\dots,c_{w_{N-K}})=
		(c_1,\dots,c_{N}),
		\]
		where $c_i\in\Cbase$. We denote the $j$-th coordinate of $c_i$ by $c_{i,j}$, so that
		$c_i = (c_{i,0},\dots,c_{i,n_0}).$ Furthermore, masking is such that 
			$$(c_{1,0},\ldots,c_{K,0},c_{K+1,0},\ldots,c_{N,0})=(u_1,\dots,u_K,w_1,\dots,w_{N-K}).$$
			The dual code of $\Clin$, denoted by $\Dlin,$
                        is a $[N(n_0+1),N(n_0+1)-Kk_0]_{\F}$
                        code. %with parity-check matrix $H$.
                        Using the above indexing, we also denote a dual
                        codeword by $h=(h_1,\dots,h_N)\in C^\perp$. When
                        puncturing $D$ at the set
                        $I_0=\set{0,\nbase+1,\dots,(K-1)(\nbase+1)}$, we again
                        abuse the underline notation and write
		\[
		\underline{h}=(\underline{h_1},\ldots,\underline{h_N}),
		\]
		where
		$\underline{h}_i=(h_{i,1},\ldots,h_{i,\nbase})\in\F^{\nbase}$ for $i\in [N]$ is the puncturing of $h_i$ at position $0$.
		Whenever convenient, we  write
		$\underline{h}=(h_i)_{i\in J},$
		with  $J=[N(\nbase+1)-1]_0\setminus I_0$.
		Let $J=J_1 \sqcup \ldots \sqcup J_{N}$ be the corresponding partition of $J$,
		we also use the notation
		\[
		\underline{h}=\big((h_i)_{i\in J_1},\,\ldots,\,(h_i)_{i\in J_{N}}\big).
		\]
		% which is consistent with the block decomposition
		% $\underline{h}=(\underline{h}_0,\ldots,\underline{h}_{N-1})$ introduced above.
  \subsubsection{Example: Summation of an Array.}\label{ex:array1}
   Let $(u_1,\dots,u_K)\in\F^K$ be an array of input values. Using 
  registers $r_u,r_w\in\F$, the computation of $u_1+\cdots+u_K$ can be 
  implemented by initializing $r_w=0$ and, for each $i\in[K]$, performing the 
  operations $r_u\gets u_i$ (memory load) and $r_w\gets r_w+r_u$ (addition). 
  We consider leakage from the successive values of $r_u$ (the $u_i$'s) and the
  successive values of $r_w$ (the partial sums, denoted by $w_i$). The associated 
  computation code for $K=4$ has parity-check matrix $\Hobs$ and generator matrix $\Gobs$ as follows:
	\[
	\Hobs =
        \kbordermatrix{
          &u_1&u_2&u_3&u_4&w_1&w_2&w_3&w_4\\
	&1 & 0  & 0  & 0  & -1 & 0 & 0 & 0 \\
	&0 & 1 & 0  & 0  & 1 & -1 & 0 & 0 \\
	&0 & 0 & 1 & 0  & 0 & 1 & -1 & 0 \\
	&0 & 0 & 0 & 1 & 0 & 0 & 1 & -1
},\quad
	\Gobs =
        \kbordermatrix{
          &u_1&u_2&u_3&u_4&w_1&w_2&w_3&w_4\\
	&1 & 0 & 0 & 0 & 1 & 1 & 1 & 1 \\
	&0 & 1 & 0 & 0 & 0 & 1 & 1 & 1 \\
	&0 & 0 & 1 & 0 & 0 & 0 & 1 & 1 \\
	&0 & 0 & 0 & 1 & 0 & 0 & 0 & 1
		}.
	\]
	Thus, the full vector that is leaked is
	$(u_1,u_2,u_3,u_4, w_1,w_2,w_3,w_4)\in\mathbb{F}^8$, where
	\[
	(u_1,u_2,u_3,u_4, w_1,w_2,w_3,w_4)
	=
	(u_1,u_2,u_3,u_4)\cdot  \Gobs.
      \]
	When masking is applied and the computation is performed in parallel on the shares, 
	we obtain $c_{w_1}= c_{u_1}$ and $c_{w_i}=c_{w_{i-1}} + c_{u_i}$. The associated
	$[8 (n_0+1), 4k_0]_\F$ product code is $\Clin=\Cobs \otimes \Cbase$. 
  \subsubsection{Example: Linear Feedback Shift Register.}\label{ex:lfsr}
	Let $s\in\F$ be a secret and let $\alpha\in\F$ be a fixed constant. A register $r_w$ is 
  initialized with $r_w=s$ and updated by $r_w=\alpha r_w$ for $i=1,\ldots,N-1$. 
  Let $w_0=s,w_1,\dots,w_{N-1}$ be the successive values of the register $r_w$. 
  The associated $[N,1]_{\F}$ computation code for $N=5$ has the following parity-check and generator matrices
\[
\Hobs =
        \kbordermatrix{
          &s&w_1&w_2&w_3&w_4\\
&\alpha & -1 & 0 & \cdots & 0 \\
&0 & \alpha & -1 & \ddots & \vdots \\
&\vdots & \ddots & \ddots & \ddots & 0 \\
&0 & \cdots & 0 & \alpha & -1
},
\quad
\Gobs =
        \kbordermatrix{
          &s&w_1&w_2&w_3&w_4\\
&1 & \alpha & \alpha^2 & \cdots & \alpha^S
}.\]
% Thus the full  vector is $(u_1,w_2,\ldots,w_{S+1})\in F^{S+1}$ can be written as
% \[
% w
% =
% (u_1,w_2,\ldots,w_{S+1})
% =
% u_1\Gobs.
% \]
Given a base code $\Cbase$,
the induced $[Nn_0,k_0]_{\mathbb{F}}$ product code  $\Clin = \Cobs \otimes \Cbase$ has codewords $c_{w_i}$ such that 
$c_{w_0} = c_s$
and
$c_{w_i} = \alpha c_{w_{i-1}}$,
 $i=1,\ldots,N-1$.

 \subsection{Repair Problem for General Linear Computations}
 For a codeword 
 $c=
 (c_1,\dots,c_{N})\in\Clin$, the repair problem  is to recover the independent input values $c_{1,0},\dots,c_{K,0}$. 
 In the context of secret sharing, this corresponds to recovering secret data. 
 In this setting, we need
 to find dual codewords with suitable properties.
 Namely, we look for a matrix $H_{Km}\in\F^{Km\times N(\nbase+1)}$ consisting of
        $Km$ dual codewords $h^{(1)},\ldots,h^{(Km)}\in \Dlin$:
		\[
		\begin{aligned}
		H_{Km}
		&=
		\begin{bmatrix}
		h^{(1)}\\
		h^{(2)}\\
		\vdots\\
		h^{(Km)}
		\end{bmatrix}
		=
		\begin{bmatrix}
		h_{1,0}^{(1)} & \cdots & h_{1,n_0}^{(1)}
		& \cdots
		& h_{N,0}^{(1)} & \cdots & h_{N,n_0}^{(1)}\\
		\vdots &  & \vdots
		& 
		& \vdots &  & \vdots\\
		h_{1,0}^{(Km)} & \cdots & h_{1,n_0}^{(Km)}
		& \cdots
		& h_{N,0}^{(Km)} & \cdots & h_{N,n_0}^{(Km)}
		\end{bmatrix}.
		\end{aligned}
		\]
		% \begin{bmatrix}
		% 	h_{0,0}^{(1)}  & \cdots & h_{0,n_0}^{(1)} 
		% 	& h_{1,0}^{(1)}& \cdots & h_{1,n_0}^{(1)} & \cdots
		% 	& h_{N-1,0}^{(1)}  & \cdots & h_{N-1,n_0}^{(1)}\\
		% 	h_{0,0}^{(2)}  & \cdots & h_{0,n_0}^{(2)} 
		% 	& h_{1,0}^{(2)} & \cdots & h_{1,n_0}^{(2)} & \cdots
		% 	& h_{N-1,0}^{(2)}  & \cdots & h_{N-1,n_0}^{(2)}\\
		% 	\vdots  & \ddots & \vdots 
		% 	& \vdots  & \ddots & \vdots & \cdots
		% 	& \vdots & \ddots & \vdots\\
		% 	h_{0,0}^{(Km)} & \cdots & h_{0,n_0}^{(Km)}
		% 	& h_{1,0}^{(Km)}  & \cdots & h_{1,n_0}^{(Km)} & \cdots
		% 	& h_{N-1,0}^{(Km)} & \cdots & h_{N-1,n_0}^{(Km)}\\
		% \end{bmatrix}.
		We first give the full-rank condition. The matrix $A_{H_{km}}\in\F^{Km\times K}$ defined by
                                        \begin{equation}
A_{H_{km}}=\begin{bmatrix}
             h_{1,0}^{(1)}&h_{2,0}^{(1)}&\cdots&h_{N,0}^{(1)}\\
            \vdots &\vdots& &\vdots\\
            h_{1,0}^{(Km)}&h_{2,0}^{(Km)}&\cdots&h_{N,0}^{(Km)}
            \end{bmatrix}\cdot \Gobs^T\label{mat:A:lin}
\end{equation}
 must satisfy (after row expansion):
        \begin{equation}
          \rank_{\B}\! (A_{H_{km}}^\B)=K\cdot m.\label{eq:full_rank:gen}
          \end{equation}
		The second condition (one dimension columns) is, for each pair $ (i,j)\in [N] \times [n_0]$:
		\begin{equation}
		\begin{aligned}
		\dim_{\B}\!\left(
		\operatorname{span}_{\B}
		\left( 
			h_{i,j}^{(1)},\ldots,h_{i,j}^{(K\cdot m)}
		\right)
		\right)
		&= 1.\\
		\end{aligned}
		\label{eq:one-dim:gen}
		\end{equation}

		\begin{theorem}\label{thm:ex:lers:gen}
			Let $\Clin= \Cobs \otimes \Cbase$ be the product code of $\Cbase$ and $\Cobs$ as  above. 
			Then the existence of a matrix $H_{K m}$ that satisfies \eqref{eq:full_rank:gen}
			 and \eqref{eq:one-dim:gen} implies the existence of an equally distributed 
			 LERS for $C$.
		\end{theorem}
		\begin{proof}
			For $c\in \Clin$ and $h=(h_1,\ldots,h_{N}) \in \Dlin,$ 
			 it holds that 
                        \[
                          \sum_{j=1}^{N}  \left(c_{j,0} h_{j,0}+\underline c_j \underline h_j^T\right)=0.
                        \]
                        With $\left(c_{1,0},\dots,c_{N,0}\right)=(c_{1,0},\dots,c_{K,0})\cdot \Gobs$, we obtain
                        \begin{align*}
                          \sum_{j=1}^{N} c_{j,0} h_{j,0}&=\left(c_{1,0},\dots,c_{N,0}\right)\cdot \left(h_{1,0},\dots,h_{N,0}\right)^T\\
                          &=(c_{1,0},\dots,c_{K,0})\cdot \Gobs\cdot \left(h_{1,0},\dots,h_{N,0}\right)^T\\
                          &=(c_{1,0},\dots,c_{K,0})\cdot \left(\left(h_{1,0},\dots,h_{N,0}\right)\cdot \Gobs^T\right)^T.
                        \end{align*}
                        From condition \eqref{eq:full_rank:gen},
                        the matrix $A_{H_{Km}}^\B$ can be reduced with $\B$-linear row operations to
                        $\textbf{1}_{Km\times Km}$, i.e.\  $A_{H_{Km}}$ can be reduced to
                        \[
                          \kbordermatrix{
                            & u_1&\dots &u_K\\
                            &\zeta_1  & \\
                            &\vdots    & \\ 
                            &\zeta_m  & \\
                            &&\ddots\\
                            &&\ddots\\
                            &&\ddots\\
                            &&&\zeta_1  & \\
                            &&&\vdots    & \\ 
                            &&&\zeta_m  & \\
                            }.
                          \]
                        Thus, for a given $j_0\in [K]$ and $i\in[m]$, there is $h^{(i)}=(h^{(i)}_{1},\dots,h^{(i)}_{N})\in \Dlin$ 
                       such that
                       \[
                         (h^{(i)}_{1,0},\dots,h^{(i)}_{N,0})\cdot\Gobs ^T
						 =(0_{1\times m},\dots,0_{1\times m},\overbrace{0,\dots,\zeta_i,\dots,0}^{j_0\text{-th block}},\dots,0_{1\times m}).
                       \]
                       Hence,
                          \[
                            (c_{1,0},\dots,c_{K,0})\cdot \left((h^{(i)}_{1,0},\dots,h^{(i)}_{N,0})\Gobs ^T\right)^T=c_{j_0,0}\zeta_i.
                          \]
						  % For each $j\in [N]$, we write $\underline{c}_j=(c_{j,1},\ldots,c_{j,\nbase})$ and
						  %       $\underline{h}^{(i)}_j=(h^{(i)}_{j,1},\ldots,h^{(i)}_{j,\nbase}).$
							From~\eqref{eq:one-dim:gen}, for $i\in[m]$, and $(j,\ell)\in [N]\times[\nbase]$, there exists  $\beta_{j,\ell}\in\F$ and  $\lambda^{(i)}_{j,\ell}\in\B$  such that 							$h^{(i)}_{j,\ell}=\lambda^{(i)}_{j,\ell}\beta_{j,\ell}$.
							Using this decomposition and applying the trace  gives
						    \begin{align*}
                             -c_{j_0,0}\zeta_i&               =\sum_{j=1}^{N} \sum_{\ell=1}^{\nbase} c_{j,\ell} \lambda^{(i)}_{j,\ell}\beta_{j,\ell},\\
                             -\Tr_{\F/\B}(c_{j_0,0}\zeta_i) &              =\sum_{j=1}^{N} \sum_{\ell=1}^{\nbase} \lambda^{(i)}_{j,\ell}\Tr_{\F/\B}(c_{j,\ell}\beta_{j,\ell}).
                          \end{align*}
			The $N\times\nbase$ leakage functions
			$c_{j,\ell}\mapsto \Tr_{\F/\B}(c_{j,\ell}\beta_{j,\ell})$ 
			are independent of $i$. Varying~$i$, the $m$ values $\Tr_{\F/\B}(c_{j_0,0}\zeta_i)$ 
			can be determined as linear combinations of these leakage functions, 
			with coefficients $\lambda^{(i)}_{j,\ell}$. 
			Since $(\zeta_1,\dots,\zeta_m)$ is a basis of $\F$, these $m$ values determine $c_{j_0,0}$. \qed
		\end{proof}

	\subsection{Subfield Subcode Analysis}
			
		We again study how to construct a matrix $H_{Km}$ that satisfies the full-rank condition \eqref{eq:full_rank:gen}
		and the one-dimensional column-span conditions \eqref{eq:one-dim:gen}. Due to Theorem \ref{thm:ex:lers:gen}, it will give 
		an equally distributed LERS for the product code $\Clin$.

		We first focus on one-dimensional
		column-span conditions~\eqref{eq:one-dim:gen}. We use the same
		approach as in the previous section. Analogously to the add case, consider the puncturing of the dual code $\Dlin$
		at the set of coordinates 
		$I_0,$
		and denote it $\underline{\Dlin}$.
		%Furthermore, denote the coordinate set of $\underline{\Dlin}$ by $J=\{0,\ldots,N(\nbase+1)-1\}\setminus I_0.$
		% Note that the sets $I_0$ and $J$ are defined analogously to the corresponding sets in
		% the previous section, but refer to the current product-code construction.
		Given the generator matrix $\Gbase$ of $\Cbase$ as in
		\eqref{eq:Gbase}, we obtain the following generator matrix for the product code
		$\Clin$:
		\[
		\Glin=\Gobs\otimes\Gbase
		=\Gobs\otimes
		\begin{bmatrix}
		-1 & \phi\\
		0 & \overline{\Gbase}
		\end{bmatrix},
		\]
		which is  a parity-check matrix of $\Dlin$. One can see that the parity-check matrix
		of the punctured code $\underline{\Dlin}$ is given by
		\[
		\overline{\Glin}=\Gobs\otimes\overline{\Gbase}.
		\]
		Furthermore, we have that
                the dimension of the punctured code $\underline{\Dlin}$ is
		\begin{equation}\label{eq:dim:punc:Dlin}
			\dim_\F(\underline{\Dlin})=N\nbase-K(\kbase-1).
		\end{equation}
		The proof of \eqref{eq:dim:punc:Dlin} can be found in Appendix \ref{appendix:proof_dim}.
			We start by analyzing which punctured dual codewords $\underline{h}^{(1)},\ldots,\underline{h}^{(Km)}\in\underline{\Dlin}$ satisfy the 
			one-dimensional column-span conditions \eqref{eq:one-dim:gen}. For each $j\in J$, let $\beta_j\in\F$ be non-zero such that the  $j$-th column of
			$H_{Km}$ is contained in the $\B$-span of $\beta_j$. Consider 
			\[
			\underline{\Dlinpb} := \left\{ \underline{h} \in \underline{D} :
			h_{j}\in \operatorname{span}_{\mathbb{B}}(\beta_{j})\ \text{ for all } j \in J \right\},
			\]
			which is a $\B$-linear code of length $N\nbase$ over the alphabet $\F$.
			Since $\beta_j\neq 0$ for all $j\in J$, there exists for each codeword $\underline h\in\underline{\Dlinpb}$ a unique vector
			$
			\underline\lambda
			\in\B^{Nn_0}
			$
			such that 
			$h_{j}=\lambda_{j}\beta_{j},$ for all $j\in J.$
			This observation motivates the definition
			\[
			\Dlambdabeta
			:=
			\left\{
			\underline \lambda \in\mathbb B^{N\nbase}
			\;\middle|\;
			\exists\,\underline h\in\underline{\Dlinpb}
			\text{ such that }
			h_{j}=\beta_{j}\lambda_{j}
			\text{ for all }j\in J
			\right\}.
			\]			
			Then $\Dlambdabeta\subseteq \B^{N\nbase}$ is a
                        $\B$-linear code over the alphabet $\B$. Let the diagonal matrix of size
                        $Nn_0 \times Nn_0$ be given by
                        $M_{\beta^{-1}}:=Diag(\beta_j^{-1})_{j\in J}.$ We can write
                        $\underline{\Dlinpb}=\Dlambdabeta
                        M_{\beta}$.
                        Then $\Dlambdabeta$ is a subfield subcode of
                        $\underline{\Dlinpb} M_{\beta^{-1}}$ (see the proof in Appendix \ref{appendix:subfieldsubcode}).
				To prove the existence of punctured dual codewords in
				$\underline{\Dlin}$ satisfying the one-dimensional constraints described in
				\eqref{eq:one-dim:gen}, it suffices to show that the subfield subcode
				$\Dlambdabeta$ has dimension at least $Km$. We obtain a lower bound on its
				dimension by applying the bound for subfield subcodes:
				$\dim_{\B}(\Dlambdabeta)
				\geq
				n(\underline{\Dlinpb})
				-
				m\,r(\underline{\Dlinpb})=
				N\nbase-mK(\kbase-1),$ where we used~\eqref{eq:dim:punc:Dlin}.
				In particular, a sufficient condition for $\dim_{\B}(\Dlambdabeta)\geq Km$ is
				\begin{equation}\label{eq:first:gen:dim:thres}
					\kbase\leq\frac{N\nbase}{Km}.
				\end{equation}
				Hence, if the upper bound~\eqref{eq:first:gen:dim:thres} for the dimension is satisfied,
				then there exist punctured dual codewords in
				$\underline{\Dlin}$ satisfying conditions
				\eqref{eq:one-dim:gen}.

		\subsubsection{Fulfilling the Full-Rank Condition.}
		Given $\underline{h}=(\underline{h_1},\ldots,\underline{h_{N}})
		\in\underline{\Dlinpb}\subseteq\underline{\Dlin},$
		we now address the full-rank condition~\eqref{eq:full_rank:gen}. In
		particular, for $\phi$ in $\Gbase$ as in~\eqref{eq:Gbase}, we require the following $K$ values
		\[
		\begin{array}{c}
			(\phi\cdot(\underline{h_1^T},\ldots,\underline{h_{N}^T}))
			\otimes \Gobs[1,\cdot],\\
			\vdots\\
			(\phi\cdot(\underline{h_1^T},\ldots,\underline{h_{N}^T}))
			\otimes \Gobs[K,\cdot]
		\end{array}
		\]
		to be linearly independent, where $\Gobs[i,\cdot]$ denotes the $i$th row of $\Gobs$. 
		\begin{definition}[Generalized extension map]
			% For  the dual  $\Dlin$ of the product code $\Cobs\otimes \Cbase$ induced by a computation described by
			% $\Cobs$, we define the associated extension map as
			For  the dual  $\Dlin$ of the product code $\Cobs\otimes \Cbase$, we define the associated extension map as
			\begin{equation*}
			\phi_{\Dlin} \colon \begin{array}[t]{rcl}
				\underline{\Dlin}& \longrightarrow &\F^K \\
				(\underline{h_1},\ldots,\underline{h_{N}})& \longmapsto &
					\left((\phi \cdot(\underline{h_{1}^T},\ldots,\underline{h_{N}^T}))\otimes \Gobs[i,\cdot]\right)_{i\in[K]}.
			\end{array} 
			\end{equation*}
		\end{definition}
		Applying this map to $\underline{h}\in\underline{\Dlinpb}$, and writing, for $j\in[\nbase],$
		\[
			h_{1,j}=\lambda_{1,j}\beta_{1,j},
			\ldots,
			h_{N,j}=\lambda_{N,j}\beta_{N,j},
		\]
		with $\underline{\lambda}=(\underline{\lambda_1},\ldots,\underline{\lambda_{N}})\in\Dlambdabeta,$
		leads to the following $\B$-linear map
		\begin{equation*}%\label{eq:phi:Lambda_gen}
		\psi_{\Dlambdabeta}:
		\begin{array}[t]{rcl}
		\Dlambdabeta&\longrightarrow& \F^K\\
		\underline{\lambda}
		&\longmapsto&
		\left(
		(\underline{\phi}
		\cdot
		(\underline{\lambda_1^T}\odot\beta_1^T,\ldots,
		\underline{\lambda_{N}^T}\odot\beta_{N}^T))
		\otimes \Gobs[i,\cdot]
		\right)_{i\in[K]}
		\end{array}
		\end{equation*}
		where the $\odot$-operation denotes the componentwise multiplication.
		If
		\[
		\dim_{\B}(\operatorname{Im}(\psi_{\Dlambdabeta}))=Km,
		\]
		then there exist
		$\underline{h}^{(1)},\ldots,\underline{h}^{(Km)}\in\Dlambdabeta$
		such that their images under $\psi_{\Dlambdabeta}$
		form a basis of $\F^K$ over $\B.$ This requires $\dim_{\B}(\Dlambdabeta)\geq Km,$ which is ensured by 
		condition~\eqref{eq:first:gen:dim:thres}.
		However, this does not imply that
		$\dim_{\B}\bigl(\operatorname{Im}(\psi_{\Dlambdabeta})\bigr)= Km.$
		Again, determining this dimension amounts to computing the dimension of a subfield
		subcode, which is difficult in general. We therefore only state a rank criterion.
		Let
		$M_{i,\beta}=\operatorname{Diag}\bigl((\beta_j)_{j\in J_i}\bigr)$ for $i\in[N]$.
		Then
		$\dim_{\B}\bigl(\operatorname{Im}(\psi_{\Dlambdabeta})\bigr)=Km$
		is equivalent to
		\begin{equation}\label{eq:criterion:rand:meth}
			\operatorname{rank}_{\B}\!\left(H^0_{\Dlambdabeta}\right)	-
			\operatorname{rank}_{\B}\!\left(H_{\Dlambdabeta}\right)=Km,
		\end{equation}
		where
		$H_{\Dlambdabeta}
		=
		\left(
		(\Gobs\otimes\overline{\Gbase})
		\operatorname{Diag}
		(M_{1,\beta},\ldots,M_{N,\beta})
		\right)_{\B}
		\in
		\B^{mK(\kbase-1)\times N\nbase}$
		is a parity-check matrix for $\Dlambdabeta$, and
		\begin{equation*}
		H^0_{\Dlambdabeta}
		=
		\left(
		\left(
		\Gobs\otimes
		\begin{bmatrix}
		\phi\\
		\overline{\Gbase}
		\end{bmatrix}
		\right)
		\operatorname{Diag}
		(M_{1,\beta},\ldots,M_{N,\beta})
		\right)_{\B}
		\end{equation*}
		is a parity-check matrix for $\ker(\psi_{\Dlambdabeta})$.    
		 Criterion~\eqref{eq:criterion:rand:meth} allows to verify that the vectors
                $\beta_i,\ldots,\beta_{N}\in\F^{N\times n_0}$ provide dual codewords in $\Dlin$ satisfying \eqref{eq:full_rank:gen}.

\subsection{Numerical Examples}
\subsubsection{Example: Summation of an Array.}
Recall the summation of an array from Example~\ref{ex:array1}. 
Then, the base code $\Cbase$ can have dimension at most
$\kbase\leq2\nbase/m$.
		% \begin{equation*}
		% 	H^0_{\Dlambdabeta}
		% 	=
		% 	\begin{bmatrix}
		% 	(\phi M_{0,\beta})_{\B}  & 0 & 0 & 0 & 
		% 	(\phi M_{4,\beta})_{\B} & (\phi M_{5,\beta})_{\B} & 
		% 	(\phi M_{6,\beta})_{\B} & (\phi M_{7,\beta})_{\B} \\
		% 	(\overline{\Gbase}M_{0,\beta})_{\B} & 0 & 0 & 0 & 
		% 	(\overline{\Gbase}M_{4,\beta})_{\B} & (\overline{\Gbase}M_{5,\beta})_{\B} & 
		% 	(\overline{\Gbase}M_{6,\beta})_{\B} & (\overline{\Gbase}M_{7,\beta})_{\B} \\
		% 	0 &(\phi M_{1,\beta})_{\B}  & 0 & 0 & 
		% 	0 & (\phi M_{5,\beta})_{\B} & 
		% 	(\phi M_{6,\beta})_{\B} & (\phi M_{7,\beta})_{\B} \\
		% 	0 & (\overline{\Gbase}M_{1,\beta})_{\B} & 0 & 0 & 
		% 	0 & (\overline{\Gbase}M_{5,\beta})_{\B} & 
		% 	(\overline{\Gbase}M_{6,\beta})_{\B} & (\overline{\Gbase}M_{7,\beta})_{\B} \\
		% 	0 & 0 & (\phi M_{2,\beta})_{\B}  & 0 & 
		% 	0 & 0 & 
		% 	(\phi M_{6,\beta})_{\B} & (\phi M_{7,\beta})_{\B} \\
		% 	0 & 0 & (\overline{\Gbase}M_{2,\beta})_{\B} & 0 & 
		% 	0 & 0 & 
		% 	(\overline{\Gbase}M_{6,\beta})_{\B} & (\overline{\Gbase}M_{7,\beta})_{\B} \\
		% 	0 & 0 & 0 & (\phi M_{3,\beta})_{\B}  & 
		% 	0 & 0 & 
		% 	0 & (\phi M_{7,\beta})_{\B} \\
		% 	0 & 0 & 0 & (\overline{\Gbase}M_{3,\beta})_{\B} & 
		% 	0 & 0 & 0 & (\overline{\Gbase}M_{7,\beta})_{\B}
		% 	\end{bmatrix}.
		% \end{equation*}

\subsubsection{Example: Linear Feedback Shift Register.}
Using the parameters from Example \ref{ex:lfsr}, the upper bound on the dimension $\kbase$ of the base code yields
$\kbase\leq N\cdot\nbase/m$.
% With $N=m$, we can have $\kbase=\nbase$. We explain this as
% follows. For a general code, let $s=c_0$ be the secret, and
% $c_1,\dots,c_{\nbase}$ the shares of $c_0$. Consider the use of the same  
% leakage map: $x\mapsto\Tr_{\F/\B} (x)$. For $j\in[\nbase]$,
% the leakage of $c_j$ after $N=m-1$ iterations is
% $\Tr_{\F/\B} (c_j),\Tr_{\F/\B} (\alpha c_j)\dots,\Tr_{\F/\B}
% (\alpha^{m-1}c_j)$. For $\alpha$ a generator of $\F$, these values
% determines $c_j$, for any $j\in[\nbase]$, and $c_0=s$ is found.
With $N=m$, we can have $\kbase=\nbase$. We explain this as follows.
For a general code, let $s=c_0$ denote the secret, and let
$c_1,\dots,c_{\nbase}$ be the shares of $c_0$. Consider using the same
leakage map for every share: $x \longmapsto \Tr_{\F/\B}(x).$
For each $j\in[\nbase]$, the leakage of $c_j$ after $m-1$ iterations is
$\Tr_{\F/\B} (c_j),\Tr_{\F/\B} (\alpha c_j)\dots,\Tr_{\F/\B}
(\alpha^{m-1}c_j)$. 
If $\alpha$ is chosen as generator of $\F,$ then these trace values uniquely
determine $c_j$. Hence every share $c_j$ can be recovered, and
consequently $c_0=s$.
% We first remark that this covers additive secret sharing, when all
% the shares sum up to secret $s$. Our second remark is that, contrarily
% to the add case, we can use the same leakage functions (same $\beta_j$'s) at each step of the computation.
We remark that this argument applies to additive
secret sharing, where the shares sum up to the secret $s$.
% In this case, the construction
% reduces to the additive secret-sharing setting.
% In contrast, if $N=\ceil{m\left(1+\frac{1}{\nbase}\right)},$
% then the uncoded choice $\kbase=\nbase+1$ is not valid. Indeed, the base code
% does not contain a dual codeword that is nonzero in the first coordinate.
% Equivalently, there exists no systematic generator matrix of $\Cbase$
% for which $\phi\neq 0$.
\section{Use of Identical Leakage Functions}\label{sec:samebetas}
In general, there are $N$ sets of $\beta_j$'s defining the leakage
functions $\Tr_{\F/\B}(\beta_j\cdot c_j)$: for each $i\in[N]$, the
$j$-th symbol of $c_i$ is associated with
$\beta_{i,j}$. Our randomized method found, with overwhelming
probability, values $\beta_{1,j},\ldots,\beta_{N,j}$ satisfying
\eqref{eq:criterion:rand:meth}, thereby yielding a linear repair
scheme. Recall, however, that these values are chosen independently
and uniformly at random. A more realistic attack model requires the
leakage functions to coincide across all $i\in[N]$, i.e.,
$\beta_{1,j}=\ldots=\beta_{N,j}=\beta_j$ for $j\in[\nbase]$.

Imposing this restriction on the simple addition computation
(introduced in \ref{ch:add_gate}), i.e., setting
$\beta_{u,j}=\beta_{v,j}=\beta_{w,j}=\beta_j$ for all $j\in[\nbase]$,
our simulations \emph{never found} a set of $\beta_j$'s satisfying
\eqref{eq:rkcond:Im:psi} under the improved dimension bound. This
outcome can be explained as follows. Suppose, for contradiction, that
\eqref{eq:rkcond:Im:psi} holds when
$\beta_{u,j}=\beta_{v,j}=\beta_{w,j}$ for all $j\in[\nbase]$. Then the
matrix in \eqref{eq:H_Lambdabeta} has $\B$-rank $2m$ greater than that
of the matrix in \eqref{eq:H_Lambdabeta0}; explicitly,
\begin{equation*}
\begin{bmatrix}
	(\phi M_{u,\beta})_{\B} & \mathbf{0} &
	(\phi M_{w,\beta})_{\B}\\
	(\overline{\Gbase}M_{u,\beta})_{\B} & \mathbf{0} &
	(\overline{\Gbase}M_{w,\beta})_{\B}
\end{bmatrix}
\end{equation*}
has $\B$-rank $m$ greater than that of
\begin{equation*}
\begin{bmatrix}
	(\overline{\Gadd_0}M_{u,\beta})_{\B} & \mathbf{0} &
	(\overline{\Gadd_0}M_{w,\beta})_{\B}
\end{bmatrix}.
\end{equation*}
Since $M_{u,\beta}=M_{w,\beta}=M_\beta$, it follows that
\begin{equation*}
\begin{bmatrix}
	(\phi M_{\beta})_{\B}\\
	(\overline{\Gbase}M_{\beta})_{\B}
\end{bmatrix}
\end{equation*}
has $\B$-rank $m$ greater than that of $(\overline{\Gbase}M_{\beta})_{\B}$.
This is precisely the rank condition~\eqref{eq:CNS:t} required for the
base code $\Cbase$. Hence, any repair scheme for the product code $C$
that uses identical leakage functions induces a repair scheme for
$\Cbase$ itself; consequently, no improvement over the base code is
possible, and the upper bound on $\kbase$ cannot be improved by
exploiting the additive relation alone.

The same argument applies to the Guruswami--Wootters~\cite{STOC:GurWoo16}
algebraic LERS for Reed--Solomon codes: under the simple addition, any
scheme using identical $\beta_j$'s for all three codewords reduces to a
LERS for the base Reed--Solomon code. In cryptographic terms, for
Shamir secret sharing, this means that exploiting the simple additive
relation requires leakage functions other than those of
\cite{STOC:GurWoo16}, if such functions exist. 
This observation also seems to extend directly to the array summation example~\ref{ex:array1}.

However, for the linear relation $\mu_u u + \mu_v v = w$ with
$\mu_u\neq\mu_v \in \F\setminus\{0,1\}$, our simulation \emph{finds},
with overwhelming probability, a set of $\beta_j$'s satisfying
\eqref{eq:criterion:rand:meth} under the improved dimension bound. It seems that the same
holds for the sum $\mu_1u_1+\cdots+\mu_Ku_K$ for $K$ 
coefficients $\mu_i\in\F\setminus\{0,1\}.$ Finally, also for the LFSR
example~\ref{ex:lfsr}, identical leakage functions can be reused at
every step of the computation.

\section{Conclusion}

Our main contribution is the construction of leakage functions, based on linear exact repair schemes, for
Massey secret sharing schemes when linear
computations are performed. More precisely, we
consider the product code $ \Clin = \Cobs \otimes \Cbase,$
where the $[\nbase,\kbase]_\F$ base code $\Cbase$ induces the
underlying Massey secret sharing scheme, and the $[N,K]_\F$ computation code
$\Cobs$ represents linear computations among the $N$ secrets, of which
$K$ are linearly independent (the inputs). Our construction
yields leakage functions for $\Clin$ with high probability whenever $\kbase \leq N\nbase/(Km),$
which improves the upper bound $\kbase \leq \nbase/m-1$
of the base case of~\cite{wcc-paper}.

It is somewhat intriguing that the (standard share-wise) addition,
which has strong security properties (Strong Non-Interfering, SNI,~\cite{CCS:BBDFGS16}) in the
probing model, presents a weakness in the LERS-based leakage
model. We also point out that, in cryptography, our
results do not hold when the shares are refreshed, since any
redundancy introduced by the computation code is then lost.
Moreover, while identical leakage functions cannot exploit simple addition, 
they can exploit more general linear relations, yielding a more realistic 
leakage attack and potentially leading to a security weakness of Massey secret sharing schemes.

Future work includes investigating multiplication. 
% when the
% threshold (corresponding to $\kbase-1$ in coding-theoretic terms)
% satisfies $ 2t \leq \nbase+1.$
Such an analysis would rely on the Schur product of codes~\cite{C:PCCX09}.

% This sample uses bibtex rather than biblatex.

% NOTES
% - Download abbrev3.bib and crypto.bib from https://cryptobib.di.ens.fr/
% - Use biblio.bib for additional references not in the cryptobib database.
%   If possible, take them from DBLP.

\printbibliography

\appendix

\section{Proof of~Eq.~\ref{eq:dim:punc:Dlin}}\label{appendix:proof_dim}

\subsection{Matrix-based Proof}

A redundant parity-check matrix of $\Clin$ (equivalently, a generator
matrix of $\Dlin$) is given by
\[
\left[
\begin{array}{c}
\mathbf{1}_N \otimes \Hbase\\
[2ex]
\hline
\Hobs\otimes \mathbf{1}_{\nbase+1}
\end{array}
\right],
\]
which contains $(2N-K)(\nbase+1)-N\kbase$ rows and
$N(\nbase+1)$ columns.
Since the submatrix formed by the last $N-K$ columns of $\Hobs$ is lower
triangular, redundant rows can be removed, such that we obtain the following
equivalent full row rank parity-check matrix:
\[
\Hlin=
\left[
\begin{array}{c|c}
\mathbf{1}_K \Hbase & \mathbf{0}_{K\times(N-K)}
\\ \hline
\star &
\begin{matrix}
\star & 0 & \cdots & 0\\
\vdots & \ddots & \ddots & \vdots\\
\vdots & & \ddots & 0\\
\star & \cdots & \cdots & \star
\end{matrix}
\end{array}
\right],
\]
which contains only $N(\nbase+1)-K\kbase$ rows.
Finally, puncturing $H$ on all positions in $I_0$ gives
\[
\left[
\begin{array}{c|c}
\mathbf{1}_K & \mathbf{0}_{K\times(N-K)}
\\ \hline
\mathbf{0}_{(N-K)\times K} & 
\mathbf{0}_{(N-K)\times(N-K)}
\\ \hline
\star &
\begin{matrix}
\star & 0 & \cdots & 0\\
\vdots & \ddots & \ddots & \vdots\\
\vdots & & \ddots & 0\\
\star & \cdots & \cdots & \star
\end{matrix}
\end{array}
\right],
\]
which contains $(N-K)$ zero rows. Removing these redundant rows yields the
parity-check matrix $\underline{\Hlin}$, which contains
$N\nbase-K(\kbase-1)$ rows. Consequently, the dimension of the punctured code $\underline{\Dlin}$ is
\begin{equation}
		\dim_\F(\underline{\Dlin})=N\nbase-K(\kbase-1).
\end{equation}
Thus, puncturing the code at the positions $I_0$ reduces its dimension by
$N-K$. Furthermore, this observation is consistent with the behavior of the
addition.

\subsection{Proof via the Abstract Duals of a Tensor Product Code}

The dual of the tensor-product code $\Cobs\otimes\Cbase$ is given by
\[
D
=
\Cobs^\perp\otimes \F^{\nbase+1}
+
\F^N \otimes\Cbase^\perp.
\]
A proof can be found in \cite[Lemma 4.1.2]{rozendaal:ths25}.
Note that this is not a direct sum. Hence, we can recover the dimension with 
\begin{align*}
\dim_{\F} (D)
&=
\dim_{\F}\left(\Cobs^\perp\otimes \F^{\nbase+1}\right)
+
\dim_{\F}\left(\F^N \otimes\Cbase^\perp \right)
-
\dim_{\F}\left(\Cobs^\perp\otimes\Cbase^\perp\right)
\\
&=
(\nbase+1)(N-K)
+
(\nbase+1-\kbase)N
-
(\nbase+1-\kbase)(N-K)
\\
&=
(\nbase+1)N-\kbase K.
\end{align*}
Puncturing at the positions in $I_0$ gives
\[
\underline{D}
=
\Cobs^\perp\otimes \F^{\nbase+1}
+
\F^N \otimes\underline{\left(\Cbase^\perp\right)}.
\]
Again, the sum is not direct. Therefore,
\begin{align*}
\dim_{\F}(\underline{D})
&=
\dim_{\F}\left(
\Cobs^\perp\otimes\F^{\nbase}
\right)
+
\dim_{\F}\left(
\F^N\otimes\underline{\left(\Cbase^\perp\right)}
\right)
-
\dim_{\F}\left(
\Cobs^\perp\otimes\underline{\left(\Cbase^\perp\right)}
\right)
\\
&=
\nbase(N-K)
+
(\nbase+1-\kbase)N
-
(\nbase+1-\kbase)(N-K)
\\
&=
\nbase N-\kbase K+K.
\end{align*}

\section{Proof of Subfield Subcode Property}\label{appendix:subfieldsubcode}

\begin{lemma}
				The code $\Dlambdabeta$ is the $\B$-subfield subcode of $\underline \Dlinpb M_{\beta^{-1}},$ i.e., 
				$$\Dlambdabeta  = \underline \Dlinpb M_{\beta^{-1}}\cap  \B^{N\nbase}.$$
			\end{lemma}
			\begin{proof}
				Observe that $\underline{\Dlinpb} \subseteq \F^{Nn_0}$ is $\F$-linear, whereas
				$\Dlambdabeta \subseteq \B^{Nn_0}$ is $\B$-linear. Let
				$\underline{\lambda}\in \Dlambdabeta$; then there exists $\underline{h}\in \underline{\Dlinpb}$
				such that $h_{j}=\lambda_{j}\beta_{j}$ for all $j\in J$.
				Since $\beta_{j} \neq 0$ for all $j\in J$, the diagonal matrix
				$M_{\beta^{-1}} = \mathrm{Diag}(\beta_1^{-1},\ldots,\beta_{N}^{-1})$
				is well-defined and invertible. Hence, for each $i\in[N]$,
				\begin{align*}
				\underline{\lambda}_i=(\lambda_{i,1},\ldots,\lambda_{i,n_0})
				=\left(\frac{h_{i,1}}{\beta_{i,1}},\ldots,\frac{h_{i,n_0}}{\beta_{i,n_0}}\right)
				=\underline{h_i}\cdot \mathrm{Diag}(\beta_i^{-1}),
				\end{align*}
				where $\underline{h_i}$ is viewed as a row vector, and right multiplication by
				$\mathrm{Diag}(\beta_i^{-1})$ corresponds to componentwise multiplication by the
				scalars $\beta_{i,j}^{-1}$. Therefore,
				\[
				\underline{\lambda}=\underline{h}\cdot M_{\beta^{-1}} \in \underline{\Dlinpb}\, M_{\beta^{-1}},
				\]
				which implies $\underline{\lambda}\in \underline{\Dlinpb}\, M_{\beta^{-1}}\cap \B^{Nn_0}$.
				For the reverse inclusion, let $g\in \underline{\Dlinpb}\, M_{\beta^{-1}}\cap \B^{Nn_0}$.
				Then there exists $\underline{h}\in \underline{\Dlinpb}$ such that
				$g=(\underline{h}_1,\ldots, \underline{h}_{N})\,M_{\beta^{-1}}$, and hence
				\begin{align*}
				g&=(\underline{h}_1,\ldots, \underline{h}_{N})\,
				\mathrm{Diag}(\beta_1^{-1},\ldots,\beta_{N}^{-1})
				=(\underline{\lambda}_1,\ldots, \underline{\lambda}_{N})\in \Dlambdabeta.
				\end{align*}
				We conclude that $\Dlambdabeta$ is precisely the $\B$-subfield subcode of the
				$\F$-linear code $\underline{\Dlinpb}\, M_{\beta^{-1}}$. \qed
				\end{proof}

\end{document}